\documentclass[11pt,pdftex,a4paper]{article}
\usepackage{amsthm}
\usepackage[margin=1in]{geometry}
\usepackage{hyperref}
\usepackage{amsmath}
\usepackage[vlined,ruled,linesnumbered,noresetcount]{algorithm2e}
\usepackage{multicol}
\usepackage{graphicx}
\usepackage{enumitem}

\usepackage{amssymb}
\usepackage{xcolor}

\usepackage{cleveref}
\crefname{section}{\S}{\S\S}

\usepackage{float}
\floatstyle{boxed} 
\restylefloat{figure}

\newtheorem{claim}{Claim}
\newtheorem{lemma}{Lemma}
\newtheorem{corollary}{Corollary}
\newtheorem{claim-subsection}{Claim}[subsection]
\newtheorem{definition}{Definition}
\newtheorem{theorem}{Theorem}

\makeatletter
\newtheorem*{rep@theorem}{\rep@title}
\newcommand{\newreptheorem}[2]{%
\newenvironment{rep#1}[1]{%
 \def\rep@title{#2 \ref{##1}}%
 \begin{rep@theorem}}%
 {\end{rep@theorem}}}
\makeatother

\newcommand{\tsref}[1]{\textsection\ref{#1}\xspace}

\newcommand{\myparagraph}[1]{\vspace{1mm}\noindent\textbf{#1.}}
\newcommand{\proofparagraph}[1]{\vspace{1mm}\noindent\textit{#1.}}

\newenvironment{proofsketch}[1][Proof sketch]{\noindent\textit{#1.} }{\hfill $\Box$\\[1mm]}
\newenvironment{remark}[1][Remark]{\noindent\textit{#1.} }{\hfill $\Box$\\[1mm]}

\setlist{nolistsep}

\newcommand{\remove}[1]{}

\DontPrintSemicolon

\newcommand{\ack}{\textit{ack}}
\newcommand{\CAS}{\mbox{\textit{CAS}}}
\newcommand{\True}{\mbox{\texttt{true}}}
\newcommand{\False}{\mbox{\texttt{false}}}

\newcommand{\op}[1]{\mbox{\sc #1}}
\newcommand{\fai}{\mbox{\sc F\&I}}
\newcommand{\readop}{\mbox{\sc Read}}
\newcommand{\writeop}{\mbox{\sc Write}}
\newcommand{\setop}{\mbox{\sc Set}}
\newcommand{\readopA}{\mbox{\sc Read1}}
\newcommand{\readopB}{\mbox{\sc Read2}}
\newcommand{\scan}{\mbox{\sc Scan}}
\newcommand{\update}{\mbox{\sc Update}}

\newcommand{\draf}{\texttt{DRAF}}
\newcommand{\sticky}{\texttt{Sticky-Bit}}
\newcommand{\drr}{\texttt{Double-Write-Register}}
\newcommand{\snapshot}{\texttt{Snapshot}}

\title{Separation and Equivalence results for the Crash-stop and Crash-recovery Shared Memory Models}
	\author{
			Ohad Ben-Baruch\\
			Ben-Gurion University, Israel\\
			ohadben@post.bgu.ac.il
			\and
			Srivatsan Ravi\\
			University of Southern California, USA\\
			srivatsr@usc.edu
		}

\begin{document}

\maketitle
\begin{abstract}
Linearizability, the traditional correctness condition for concurrent data structures is considered insufficient for the non-volatile shared memory model where processes recover following a crash. For this crash-recovery shared memory model, strict-linearizability is considered appropriate since, unlike linearizability, it ensures operations that crash take effect prior to the crash or not at all. 
This work formalizes and answers the question of whether an implementation of a data type derived for the crash-stop shared memory model is also strict-linearizable in the crash-recovery model.

This work presents a rigorous study to prove how helping mechanisms, typically employed by non-blocking implementations, is the algorithmic abstraction that delineates linearizability from strict-linearizability. Our first contribution formalizes the crash-recovery model and how explicit process crashes and recovery introduces further dimensionalities over the standard crash-stop shared memory model. We make the following technical contributions that answer the question of whether a help-free linearizable implementation is strict-linearizable in the crash-recovery model: (i) we prove surprisingly that there exist linearizable implementations of object types that are help-free, yet not strict-linearizable; (ii) we then present a natural definition of help-freedom to prove that any obstruction-free, linearizable and help-free implementation of a total object type is also strict-linearizable.
The next technical contribution addresses the question of whether a strict-linearizable implementation in the crash-recovery model is also help-free linearizable in the crash-stop model. To that end, we prove that for a large class of object types, a non-blocking strict-linearizable implementation cannot have helping. 
Viewed holistically, this work provides the first precise characterization of the intricacies in applying a concurrent implementation designed for the crash-stop (and resp. crash-recovery) model to the crash-recovery (and resp. crash-stop) model.     
\end{abstract}
\newpage
\tableofcontents
\section{Introduction}
\label{sec:intro}
Concurrent data structures for the standard volatile shared memory model typically adopt linearizability as the traditional safety property~\cite{Her91}. However, in the non-volatile shared memory model where processes \emph{recover} following a crash, linearizability is considered insufficient since it allows object operations that crash to take effect anytime in the future. In the crash-recovery model~\cite{golab15}, linearizability is strengthened to force crashed operations to take effect before the crash or not take effect at all, so-called \emph{strict-linearizability}~\cite{AF03}. While there exists a well-studied body of linearizable data structure implementations in the crash-stop model~\cite{HS08-book}, concurrent implementations in the crash-recovery model are comparatively nascent. Consequently, it is natural to ask: under what conditions is a linearizable implementation in the crash-stop also strict-linearizable in the crash-recovery model? 

Non-blocking implementations in the crash-stop model employ \emph{helping}: 
i.e., apart from completing their own operation, processes perform additional
work to help \emph{linearize} concurrent operations and make progress. This helping mechanism enables an operation invoked by a process $p_i$ to be linearized by the event performed of another process $p_j$, but possibly after the crash of $p_i$. However, strict-linearizability stipulates that the operation invoked by $p_i$ be linearized before the crash event. Intuitively, this suggests that linearizable implementations that are \emph{help-free} must be strict-linearizable (also conjectured in \cite{golab15}). This work formalizes and answers this precise question: whether a help-free implementation of a data type derived for the crash-stop model can be used \emph{as it is} in the crash-recovery model.

Precisely answering this question necessitates the formalization of the crash-recovery shared memory model. Explicit process crashes introduces further dimensionalities to the set of executions admissible in the crash-recovery model over the well formalized crash-stop shared memory~\cite{AW98}. Processes may crash on an individual basis, i.e., an event in the execution corresponds to the crash of a single process (we refer to this as the \emph{individual crash-recovery model}). An event may also correspond to $m$ ($1 < m \leq n$), process crashes where $n$ is total number of processes participating in the concurrent implementation (when $m= n$ it is the \emph{full-system crash-recovery model}). Following a crash event in this model, the local state of the process is reset to its initial state when it recovers and restarts an operation assuming the \emph{old identifiers crash-recovery model} (and resp. \emph{new identifiers crash-recovery model}) with the original process identifier (and resp. new process identifier). Our contributions establish equivalence and separation results for crash-stop and the identified crash-recovery models, thus providing a precise characterization of the intricacies in applying a concurrent implementation designed for thecrash-stop model to the crash-recovery model, and vice-versa. 
%crash-stop (and resp. crash-recovery) model to the crash-recovery (and resp. crash-stop) model.     
%\myparagraph{Roadmap} 

\subsection{Contributions}
%
%Our first contribution is defining the crash-recovery model and its characteristics.
First, we define the crash-recovery model and its characteristics. We show that there exist sequential implementations of object types in the crash-stop model that as is have inconsistent sequential specifications in the old identifiers crash-recovery model.

We then consider how data structures use helping in the crash-stop model by adopting the definitions of \emph{linearization-helping}~\cite{help15} and \emph{universal-helping}~\cite{DBLP:jattiya-help}. When considering an execution with two concurrent operations, 
the linearization of these operations dictates which operation take effect first. The definition of linearization-helping considers a specific event $e$, in which it is \emph{decided}
which operation is linearized first. In an implementation that does not have linearization-helping, $e$ is an event by the process whose operation is decided to be the
one that comes first. 
%\textcolor{blue}{Need to fix below}
Universal-helping requires that the progress of some processes eventually
ensures that all pending invocations are linearized, thus forcing a process to ensure concurrent operations of other processes are eventually linearized.
%More specifically, an implementation has universal-helping if for every execution $E$, for every long enough extension of it, all pending operations in $E$ are linearized, and that linearization order does not change due to any future steps.

The first technical contribution of this paper is proving that some pairs of conditions are incomparable. That is, an implementation can satisfy exactly one of them, both, or none.
\begin{itemize}
    \item linearization-helping vs. universal-helping
    \item strict-linearizability vs. linearization-helping
    \item strict-linearizability vs. universal-helping
\end{itemize}

The second technical contribution is to show that under certain restrictions there is a correlation between some of the above pairs.
\begin{itemize}
    \item Restricting the definition of linearization-helping to be \emph{prefix-respecting}, we prove that linearization-help free implies strict-linearizability.
    More specifically, any \emph{obstruction-free} implementation of a \emph{total} object type that is linearizable and has no linearization-helping in the crash-stop model is also strict-linearizable in the new identifiers individual crash-recovery model (Lemma~\tsref{lm:prefix}).
    \item We prove that any non-blocking implementation of an \emph{order-dependent} type that is strict-linearizable in the crash-recovery model has no universal-helping in the crash-stop model (Lemma~\tsref{lm:oo-eq}).
\end{itemize}

\myparagraph{Roadmap} The contributions in this paper are structured as follows: \tsref{sec:prel} introduces the crash-stop shared memory model and other preliminaries. \tsref{sec:model} presents our characterization of the dimensionalities of the crash-recovery shared memory model. \tsref{sec:help} recalls universal-helping, linearization-helping, valency-helping and presents new results on implementations satisfying these definitions. \tsref{sec:SL-LH} discuss the correlation between strict-linearizable implementations and linearization-helping. \tsref{sec:HelpFree-SL} proves that help-freedom does not implies strict-linearizability in general, but under a natural definition of help-freedom it does follows. \tsref{sec:SL-UH} proves that strict-linearizability and universal-helping are independent. However, for a large class of objects, strict-linearizability implies universal-help freedom. Finally, \tsref{sec:SL-VH} discuss the relation between strict-linearizability and valency-helping. The paper is concluded
with a short discussion in \tsref{sec:disc}.

%\tsref{sec:tc2} presents our results that answer the question: is a linearizable help-free implementation of a data type in the crash-stop model also strict-linearizable in the crash-recovery model? \tsref{sec:tc1} presents our results to the converse question: does a strict-linearizable implementation in the crash-recovery model imply a help-free linearizable implementation in the crash-stop model? Appendix~\tsref{sec:appmain} presents results on order-dependent types and the closely related exact-order types~\cite{help15}.  

\subsection{Related work}
Strict-linearizability was proposed by Aguilera et al.~\cite{AF03} who show that it precludes wait-free implementations of multi-reader single-writer registers from single-reader single-writer registers. \cite{golab15} showed that this is in fact possible with linearizability thus yielding a separation between the crash-stop and crash-recovery models. That helping mechanisms, typically employed by non-blocking implementations, is the algorithmic abstraction that may delineate linearizability from strict-linearizability was also conjectured in \cite{golab15}. This is the first work to conclusively answer this question by providing the first precise characterization of the intricacies in applying a shared memory concurrent implementation designed for the crash-stop (and resp. crash-recovery) model to the crash-recovery (and resp. crash-stop) model.

Censor-Hillel et al.\cite{help15} formalized linearization-helping and showed that without it, certain objects called \emph{exact-order} types lack wait-free
linearizable implementations (assuming only read, write, compare-and-swap, fetch-and-add primitives) in the standard crash-stop shared memory model. Universal-helping and valency-helping were defined by Attiya et al.~\cite{DBLP:jattiya-help}. Informally, it was shown in \cite{DBLP:jattiya-help} that a non-blocking $n$-process linearizable implementation of a queue or a stack with universal-helping can be used to solve $n$-process consensus. This result was also extended to \emph{strong-linearizability}~\cite{golab-strong} which requires that once an operation is linearized, its linearization order cannot be changed in the future. The definition of strong-linearizability does bear resemblance with the proposed helping definitions in \cite{help15,DBLP:jattiya-help}; however, it is defined as restriction of linearizability and is incomparable to helping. Indeed, \cite{help15} makes the observation that strong-linearizability is incomparable with linearization-helping. The results in this paper study the implications of the universal, linearization and valency helping definitions for strict-linearizability in the crash-recovery, which has not been studied carefully thus far.

\section{Crash-stop Model}
\label{sec:prel}
This section presents the preliminaries of the standard \emph{volatile} shared memory model in which processes stop participating following a crash.

\myparagraph{Processes and shared memory}
We consider an asynchronous shared memory system in which a set of $\mathbb{N}$ processes communicate by applying \emph{operations} on shared \emph{objects}.
Each process $p_i;i\in \mathbb{N}$ has an unique identifier and an initial state.
An object is an instance of an \emph{abstract data type} which specifies a set of operations that provide the only means to
manipulate the object.
An \emph{abstract data type} $\tau$ defines a set of operations, a set of responses, a set of states, an initial state and a transition relation that determines, for each state
 and each operation, the set of possible  resulting states and produced responses~\cite{AFHHT07}. 
 %Here, $(q,\pi,q',r) \in \delta$ implies that when
 %an operation $\pi \in \Phi$ is applied on an object of type $\tau$
 %in state $q$, the object may move to state $q'$ and return a response $r$.
 We consider only \emph{deterministic} types: when
 an operation $\pi$ is applied on an object of type $\tau$
 in state $q$, there is exactly one state $q'$ to which the object can move to and exactly one matching response $r$.
 An object type is \emph{total} if any operation of the object type applied by \emph{any} process (we assume that the identifier of the process does not matter) is well-defined for every object state.
 An example for an object that is not total can be drawn by restricting known objects. For example, consider a stack where process $p$ is allowed to perform $Push(i)$ only after completing $Push(i-1)$. In such object, a process is not allowed to invoke $Push(2)$ as its first operation. 
 %\emph{Mutual exclusion} is another example of an object type that is not total. 
 However, to the best of our knowledge, most objects types are total.

An \emph{implementation} of an object type (sometimes we just say object) $\tau$ provides a specific data-representation of $\tau$ by 
applying \emph{primitives} on a set of 
shared \emph{base objects} $b_1, b_2, \ldots$, each of which is assigned an initial value and a set of algorithms $I_1(\tau),\ldots , I_n(\tau)$,
one for each process.
We assume that the primitives applied on base objects are \emph{deterministic}.
A primitive is a generic \emph{read-modify-write} (\emph{rmw}) procedure applied to a base object~\cite{G05,Her91}.
It is characterized by a pair of functions $\langle g,h \rangle$:
given the current state of the base object, $g$ is an \emph{update function} that
computes its state after the primitive is applied, while $h$ 
is a \emph{response function} that specifies the outcome of the primitive returned to the process.
Let $e$ be an event issued by some transaction that applies the rmw primitive $\langle g,h \rangle$ to a base object $b$ after an execution $E$. Let $v$ be the value of $b$ after $E$.
Now, $e$ atomically performs the following: it updates the value of $b$ to the value specified  by the function $g$ and returns a response specified by the function $h$.
%\ohad{what is $\iota \in I$ ? should it be $g(v,e)$? and a response $h(v,e)$?}
%A RMW primitive is \emph{trivial} if it never changes the value of the base object to which it is applied.
%Otherwise, it is \emph{nontrivial}.

%\myparagraph{Processes.}
%An RMW primitive $\langle g,h \rangle$ is \emph{conditional} if there exists $v$, $w$ such that
%$g(v,w)=v$ and there exists $v$, $w$ such that
%$g(v,w)\neq v$~\cite{cond-04}.
%For \emph{e.g}, \emph{compare-and-swap (CAS)}
%and \emph{load-linked/store-conditional (LL/SC} are nontrivial conditional RMW primitives
%while \emph{fetch-and-add} is an example of a nontrivial RMW primitive that is not conditional.

\myparagraph{Executions and configurations}
%An \emph{execution fragment} is a finite or infinite sequence of \emph{events}.
An \emph{event} of a process $p_i$ in the crash-stop model (sometimes we say \emph{admissible step} of $p_i$)
is an invocation or response of an operation performed by $p_i$ or a 
rmw primitive applied by $p_i$ to a base object
along with its response. 
%$r$ and write $(b, \langle g,h\rangle, r,i)$).
A \emph{configuration} specifies the value of each base object and 
the state of each process.
The \emph{initial configuration} is the configuration in which all 
base objects have their initial values and all processes are in their initial states.

An \emph{execution fragment} is a (finite or infinite) sequence of events.
An \emph{execution} of an implementation $I$ is an execution
fragment where, starting from the initial configuration, each event is
issued according to $I$ and each response of a rmw event matches the state of $b$ resulting from all preceding events.
An execution $E\cdot E'$, denoting the concatenation of $E$ and $E'$,
is an \emph{extension} of $E$ and we say that $E'$ \emph{extends} $E$.
Let $E$ be an execution fragment.
For every process identifier $k$,
$E|k$ denotes the subsequence of $E$ restricted to events of
process $p_k$.
%Let $\ms{inv}(op_m)$ denote the invocation of some t-operation $op_m$.
If $E|k$ is non-empty,
%and $H|k \neq \ms{inv}(op_m) \cdot A_k \vee \TryC_k() \cdot C_k$,
we say that $p_k$ \emph{participates} in $E$, else we say $E$ is \emph{$p_k$-free}.
An operation $\pi$ \emph{precedes} another operation $\pi'$ in an execution
$E$, 
denoted $\pi \rightarrow_{E} \pi'$, 
if the response of $\pi$ occurs before the invocation of $\pi'$ in $E$.
Two operations are \emph{concurrent} if neither precedes
the other. 
%An execution is \emph{rw-sequential} if every invocation of a read or write operation is immediately followed by a response event.
An execution is \emph{sequential} if it has no concurrent 
operations. 
Two executions $E$ and $E'$ are \emph{indistinguishable} to a set $\mathcal{P}$ of processes, if
for each process $p_k \in \mathcal{P}$, $E|k=E'|k$.
An operation $\pi_k\in ops(E)$ is \emph{complete in $E$} if
it returns a matching response in $E$.
%ends with a response event.
Otherwise we say that it is \emph{incomplete} or \emph{pending} in $E$.
We say that an execution $E$ is \emph{complete} if every invoked operation is complete in $E$.

\myparagraph{Well-formed executions} In the crash-stop model, we assume that executions are \emph{well-formed}:
no process invokes a new operation before
the previous operation returns.
Specifically, we assume that for all $p_k$, $E|k$ begins with the invocation of an operation, is
sequential and there is no event between a matching response event and the subsequent following invocation.

\myparagraph{Safety property: Linearizability}
A \emph{history} $H$ of an execution $E$ is the subsequence of $E$ consisting of all
invocations and responses of operations.
Histories $H$ and $H'$ are \emph{equivalent} if for every process
$p_i$, $H|i=H'|i$.
A complete history $H$ is \emph{linearizable} with 
respect to an object type $\tau$ if there exists
a sequential history $S$ equivalent to $H$ such that
(1) $\rightarrow_{H}\subseteq \rightarrow_S$ and
(2) \emph{$S$ is consistent with the sequential specification of type $\tau$}.
A history $H$ is linearizable if it can be
\emph{completed} (by adding matching responses to a subset of
incomplete operations in $H$ and removing the rest)
to a linearizable history~\cite{HW90,AW04}.

%\myparagraph{Liveness}

\section{Characterization of the Crash-recovery Model}
\label{sec:model}
%\subsection{Atomic memory model}
%
\myparagraph{Processes and non-volatile shared memory}
We extend the crash-stop model defined in \tsref{sec:prel} by allow any process $p_i$ to \emph{fail by crashing}; following a crash, process $p_i$ does not take any steps until the invocation of a new operation. Following a crash, the state of the shared objects remains the same as before the crash; however, the local state of crashed process is set to its initial state.

%\ohad{the use of $\mathbb{N}$ is confusing. If we start with N processes, and one fails and get a new id, how can it be this id is $k \in \mathbb{N}$ and it is both a new identifier and at most N process are in the system}

\myparagraph{Executions and configurations}
An \emph{event} of a process $p_i$ in the crash-recovery model is any step admissible in the crash-stop model as well as a special $\bot_{\mathbb{P}}$ crash step; $\mathbb{P}$ is a set of process identifiers.
The $\bot_{\mathbb{P}}$ step performs the following actions: 
%\begin{enumerate}
(i) for each $i\in \mathbb{P}$, the local state of $p_i$ set to its initial state, 
%(ii) for each $p_i\in \mathbb{P}$,
(ii) the execution $E_1\cdot \bot_{\mathbb{P}}\cdot E_2$; $E_2$ is $\mathbb{P}$-free, is indistinguishable to every process $j \notin \mathbb{P}$ from the execution $E_1\cdot E_2$. In other words, processes are not aware to crash events.
%\end{enumerate}
%\textcolor{blue}{Must state the shared object state stays the same even after the crash}
%
%\myparagraph{Persistent memory model}
%\textcolor{blue}{We assume that a process is recovered by executing a new operation on some object, not necessarily the object on which it was operating when the crash occurred. That is, we do not assume a model in which upon a recovery a process first executes a recovery function, which is in charge to first fix any inconsistencies caused by the crash before proceeding to the next operation.}
%This work assumes the \emph{atomic} persistency memory model

\myparagraph{Process crash model}
We say that an execution $E$ is admissible in the \emph{individual crash-recovery model} if for any event $\bot_{\mathbb{P}}$ in $E$, $|\mathbb{P}|=1$.
%We say that an execution $E$ is admissible in the \emph{group-process crash model} if for any event $\bot_{\mathbb{S}}$ in $E$, $|S|>1$. 
If $|\mathbb{P}|=\mathbb{N}$, we refer to it as the \emph{system-wide crash-recovery model}.
We say that an implementation $I$ is admissible in the \emph{individual crash-recovery model} (resp. \emph{system-wide crash-recovery model}) if every execution of $I$ is admissible in the \emph{individual crash-recovery model} (resp. \emph{system-wide crash-recovery model}).

%We say that an implementation $I$ is admissible in the \emph{individual-process crash model} if there exists an execution of $I$ that is admissible in the \emph{individual-process crash model}.
%We say that an implementation $I$ is admissible in the \emph{system-wide process crash model} if every execution of $I$ is admissible in the \emph{system wide crash model}.

\myparagraph{Safety property: Strict-Linearizability}
A history $H$ is \emph{strict-linearizable} with 
respect to an object type $\tau$ if there exists
a sequential history $S$ equivalent to $H^c$, a \emph{strict completion of $H$}, such that
(1) $\rightarrow_{H^c}\subseteq \rightarrow_S$ and
(2) \emph{$S$ is consistent with the sequential specification of $\tau$}.

A strict completion of $H$ is obtained
from $H$ by inserting matching responses for a subset of pending operations after the operation’s
invocation and before the next crash step (if any), and finally removing any remaining pending
operations and crash steps.

\myparagraph{Liveness}
An object implementation is \emph{obstruction-free} if for any execution $E$ and any pending operation $\pi_i$ by process $p_i$,
$\pi$ returns a matching response in $E\cdot E'$ or crashes where $E'$ is the complete \emph{solo-run} ($E'$ only contains steps of $p_i$ executing $\pi$) execution fragment of $\pi$ by $p_i$.
%and any process process $p_i$ that invokes an operation $\pi$ immediately after $E$, $\pi$ returns a matching response in $E\cdot E'$ where $E'$ is the complete \emph{step-contention free} ($E'$ only contains steps of $p+i$ executing $\pi$) execution fragment of $\pi$ by $p_i$. 
An object implementation is \emph{non-blocking} if in every execution, at least
one of the \emph{correct} processes completes its operation in a finite number of steps or it crashes. An object implementation is \emph{wait-free} if in every execution, every \emph{correct} process completes its operation within a finite number of its own steps or crashes. Obviously, liveness in the crash-stop model is identical to the above without the option of process crashing.

\myparagraph{Old identifiers crash-recovery model}
Consider an execution $E$ and a process $p_i$ that crashes in $E$.
We say that an execution $E$ is admissible in the \emph{old identifiers crash-recovery model} if for any process $p_i$ and any event $\bot_{\mathbb{P}}$ in $E$ such that $i\in \mathbb{P}$, $p_i$ takes its first step in $E$ after the crash by invoking a new operation.

%\myparagraph{A fundamental limitation of the old identifiers crash-recovery model.}
%

\myparagraph{New identifiers crash-recovery model}
We say that an execution $E$ is admissible in the \emph{new identifiers crash-recovery model} if for any process $p_i$ and any event $\bot_{\mathbb{P}}$ in $E$ such that $i\in \mathbb{P}$, process $p_i$ no longer takes steps following $\bot_{\mathbb{P}}$ in $E$. 
Note that even in this model, there are at most $\mathbb{N}$ \emph{active} processes in an execution, i.e., processes that have not crashed.

\vspace{1mm}
As we next prove, given an implementation in the crash-stop model, using it as is in the old identifiers crash-recovery model may result a sequential execution in which a process returns an invalid response. Therefore, it is not trivial to transform an implementation from the crash-stop model to the old identifiers crash-recovery model.
On the other hand, any execution in the new identifier crash-recovery model with crash events is indistinguishable to all non-crashed processes from an execution in the crash-stop model in which every crashed process simply halts, and vice-versa. Thus, and by abuse of notation, we can consider the same execution in both models in the context of deriving proofs for a given implementation.

For this reason, all results in this work concern the new identifiers crash-recovery model, thus we do not state the model explicitly. We note that all impossibility results in this paper holds also for the old identifiers crash-recovery model. This stems from the fact that given an execution in the new identifiers crash-recovery model, it can be seen as an execution in the old identifiers crash-recovery model when $\mathbb{N}$, the total number of processes in the system, is larger then the number of processes taking steps in the execution.

\begin{lemma} \label{lemma1}
There exists a sequential implementation $A$ of an object type $\tau$ in the crash-stop model providing \emph{sequential} liveness, such that $A$ is not consistent with the sequential specification of $\tau$ in the old identifiers system-wide crash-recovery model.
\end{lemma}
To prove the claim we present an implementation $A$ of type $\tau$, and construct an execution in the old identifiers system-wide crash-recovery model, such that process $p_i$ invokes and completes in a crash-free manner an operation $\sigma$; however $\sigma$ returns a response that is not consistent with the sequential specification of $\tau$ (sequential liveness in the lemma statement simply requires that a process running sequentially will complete its operation with a matching response).
Notice that the lemma holds for the more restricted system-wide crash model, hence it holds for the individual crash model as well.
Moreover, the lemma holds even for the weak sequential liveness progress condition, where no concurrency is allowed.

\begin{algorithm}[H]
        
        \footnotesize
		\begin{flushleft}
		    \textbf{Private variables}: SWSR register $R[i]$ initially 0
		\end{flushleft}
    
        \begin{procedure}[H]
			\caption{() \small $\pi(args)$}
			
			\lIf{$R[i] \neq 0$} {\KwRet $\aleph$} \label{algA:read-Ri}
			$R[i] \leftarrow 1$ \label{algA:set-Ri} \;
			Proceed to execution of $\pi(args)$ as in implementation $B$. \;
			Before returning set $R[i] := 0$ \;
		\end{procedure}
		
		\caption{Algorithm $A$. Code for process $p_i$}
    \end{algorithm}

\begin{proof}
    Let $B$ be a sequential implementation of an object type $\tau$. Assume for simplicity no operation returns the value $\aleph$. Implementation $A$ is obtained by `wrapping' each operation of $B$ with an extra code. Each process $p_i$ has a single-writer single-reader register $R[i]$ initialized to 0. At the beginning of each operation $p_i$ first reads $R[i]$. If it contains value different then 0 then $p_i$ stops and returns $\aleph$. Otherwise, $p_i$ writes 1 to $R[i]$, continue with the execution of $\pi$ as in implementation $B$, and finally before returning writes 0 back to $R[i]$.
    
    Clearly, algorithm $A$ is a sequential implementation of object $\tau$ in the crash-stop model, since any process $p_i$ takes the same steps as algorithm $B$ except for the extra reads and writes of $R[i]$. Since $R[i]$ is initialized to $0$, and as $p_i$ set it back to $0$ before completing any operation, $R[i]$ value is $0$ at the beginning of each operation.
    Consider the following execution in the crash-recovery model. Process $p_i$ invokes an operation $\pi$ and the system crash right after $p_i$ sets $R[i] \leftarrow 1$ in line \ref{algA:set-Ri}. Next, the system recovers $p_i$ by invoking a new operation $\sigma$, and let $p_i$ complete it in a crash-free manner.
    %Denote the resulted execution by $E$.
    %Notice that such an execution exists in any implementation which guarantee weak obstruction-free progress property, as $p$ is the only process to take steps in $E$. Moreover,
    Since $p_i$ crashed after setting $R[i]$ to 1, $p_i$ reads 1 in line \ref{algA:read-Ri} when executing $\sigma$, and thus returns $\aleph$.
\end{proof}

%Henceforth, all results in this work concern the new identifiers crash-recovery model, thus we do not state the model explicitly.
%Any execution in the new identifier crash-recovery model with crash events is indistinguishable to all non-crashed processes from an execution in the crash-stop model in which every crashed process simply halts. Thus, and by abuse of notation, we consider the same execution in both models in the context of deriving proofs for a given implementation.

%
%\subsection{Explicit epoch persistency memory model}

\section{Process Helping}
\label{sec:help}

In this section we present the various variants of helping based on previous works \cite{help15, DBLP:jattiya-help}. We then show that linearization-helping and universal-helping are not comparable, i.e., one does not implies the other.

\myparagraph{Linearization-helping  (\cite{help15}, rephrased)}
We say that $f$ is a \emph{linearization function} over a set of histories $\mathcal{H}$, if for every $H\in \mathcal{H}$, $f(H)$ is a linearization of $H$. We say that operation $\pi_1$ is decided before $\pi_2$ in $H$ with respect to $f$ and a set of histories $\mathcal{H}$, if there exists no $S\in \mathcal{H}$ such that $H$ is a prefix of $S$ and $\pi_2 <_{f(S)} \pi_1$.
Throughout the paper, the binary relation $<$ is used to denote that the linearization of one operation precedes another.

A set of executions $\mathcal{E}$ is \emph{linearization-help free} if there exists a linearization function $f$ over $\mathcal{E}$, such that for any two operations $\pi_1, \pi_2 \in E$ and a single step $\gamma$ such that $E\cdot \gamma \in \mathcal{E}$, it holds that if $\pi_1$ is decided before $\pi_2$ in $E\cdot \gamma$ and $\pi_2$ is not decided before $\pi_1$ in $E$, then $\gamma$ is a step of $\pi_1$ by the process that invoked $\pi_1$.
We say that an implementation is \emph{linearization-help free} if the set of admissible histories is linearization-help free.

\myparagraph{Universal-helping (\cite{DBLP:jattiya-help}, rephrased)}
for simplicity and without loss of generality, for the purposes of defining universal-helping, we assume that the first step of every operation is to publish its signature (i.e., the operation type and its operands). 
Consider a linearizable implementation $A$ of an object type $\tau$ and a function $t:\mathbb{N}\mapsto \mathbb{N}$. Then, $A$ has \emph{t-universal-helping} (when $t$ is clear from the context, we leave it out) if for every finite execution $E\cdot E'$ such that some process completes $t(n)$ or more operations in $E'$ whose invocations are contained in $E'$, there is a linearization of $E\cdot E'$ satisfying the following conditions:
\begin{itemize}
    \item linearization of $E\cdot E'$ contains every operation that is incomplete in $E$
    \item for every extension $E''$, the execution $E\cdot E' \cdot E''$ has a linearization such that the linearization of $E\cdot E'$ is the same
\end{itemize}
Otherwise, we say that $A$ is \emph{universal-help free}.

\vspace{1mm}
%\subsection{Linearization-helping vs. Universal-helping}
\cite{DBLP:jattiya-help} proved that universal-helping implies linearization-helping. However, a careful inspection of the proof reveals an implicit assumption on the object type was made. Roughly speaking, \cite{DBLP:jattiya-help} conclude that if a pending operation needs to be linearized by steps of other process due to universal helping then this implies linearization-helping. Although it is the case for many objects, the key point to consider is that universal-helping and linearization-help free definitions requires the existence of a linearization function satisfying specific conditions. Exploiting this flexibility we prove an implementation has universal-helping using some linearization function, while proving it is also linearization-help free using a different linearization function.

\newcommand{\countt}{\textsf{count}}

\begin{claim}
\label{cl:help1}
    There exists a wait-free strict-linearizable implementation $A$ of an object type $\tau$ in the individual crash-recovery model, such that $A$ has universal-helping and it is linearization-help free in the crash-stop model.
\end{claim}

\begin{proof}
A \texttt{$k$-bounded Counter} $\tau$ is an object type supporting a single operation \op{Fetch\&Increment} (\fai). The initial value of $\tau$ is 0. A \fai\ operation $\pi$ applies to $\tau$ with value $l$ changes the value of the object to $l+1$ and returns $l$ if $l<k$, otherwise $l=k$ and $\pi$ returns $k$ without changing $\tau$'s value.
Algorithm \tsref{alg:FAI} is a wait-free implementation of a \texttt{$k$-bounded Counter} using \CAS\ (compare-and-swap) primitive.

\proofparagraph{Linearization-help free}
To prove Algorithm \tsref{alg:FAI} is linearization-help free, it is enough to present a linearization function such that each operation $\pi$ is linearized at a step by the process performing it. Consider the following linearization function: for any execution $E$ and an operation $\pi \in ops(E)$, $\pi$ is linearized at its successful \CAS\ event, if such exists. Otherwise, if $\pi$ reads the value $k$ in line~\ref{line:one}, this is the linearization point of $\pi$.

\proofparagraph{Universal-helping}
We present a linearization function satisfying the conditions of the definition for universal-helping.
%linearization of $E\cdot E'$ contains every operation that incomplete in $E$, the execution $E\cdot E' \cdot E''$ has a linearization such that the linearization of $E\cdot E'$ is the same for every extension $E''$ ($E\cdot E'$ is an execution of Algorithm~\tsref{alg:FAI}). Consider the following linearization points:
An operation $\pi$ performing a successful \CAS\ operation is linearized at the point of the \CAS. In addition, once some operation $\pi$ changes the value of \countt\ to $k$, any other pending operation that is yet to be linearized, is linearized (in an arbitrary order) immediately after $\pi$. From that point on, any new invoked operation is linearized on its first step.
Let $\pi$ be a pending \fai\ operation in an execution $E$. The following holds -- either $\pi$ already have a linearization point in $E$; or that if any other process completes $k$ operations starting from $E$ then $\pi$ have a linearization point. Moreover, the assignment of linearization points is the same for any extending execution.

\proofparagraph{Strict-Linearizability}
The linearization function presented above to prove Algorithm \tsref{alg:FAI} is linearization-help free, also proves it is strict-linearizable. Any operation $\pi$ is linearized on a step by its owner. Hence, in case process $p_i$ crashes while executing an operation $\pi$, either the operation was linearized before the crash of $p_i$, or that $\pi$ has no linearization point in any extending execution.
\end{proof}

\begin{algorithm}[H]
        
%        \nonl
%        \removelatexerror
        \footnotesize
        
		\begin{flushleft}
		    \textbf{Shared variables:}
		        $\countt := 0$ \\ 
		\end{flushleft}
    
        \begin{procedure}[H]
			\caption{() \small $\op{Fetch\&Increment} ()$}
			
			\While{true} {
			    $val := \countt$ \; \label{line:one}
			    \lIf {$val = k$} {
			        \KwRet $k$
			    }
			    \uIf {$\CAS (\countt,val,val+1)$} {
			        \KwRet $val$ \;
			    }
			}
		\end{procedure}
		
		\caption{\texttt{$k$-bounded Counter}. Code for process $p_i$}
		\label{alg:FAI}
    \end{algorithm}

%\textcolor{blue}{We need to clarify the relationship between the definitions in detail}
\begin{lemma}
\label{lm:help1}
There exists an implementation of a data type that satisfies universal-helping (resp. linearization-helping), but does not satisfy linearization-helping (resp. universal-helping). 
\end{lemma}

\begin{proofsketch}
The bounded counter implementation from Claim~\tsref{cl:help1} gives the proof for one direction of Lemma~\tsref{lm:help1}. To prove the other direction, we observe that any implementation in which only a subset of the operations are getting help has linearization-helping but no universal-helping. For example, in the Binary Search Tree implementation of Ellen et al.~\cite{EllenFRB10} update operations help each other to complete, while find operations do not complete any incomplete update operation. Therefore, intuitively two update operations can prove linearization-helping. However, an infinite sequence of find operations do not complete any pending update operation; thus denying universal-helping since every pending operation must be eventually linearized. \cite{DBLP:jattiya-help} also describes an implementation satisfying linearization-helping, but not universal-helping.
\end{proofsketch}

The main results in this work focus on universal and linearization-helping. However, \cite{DBLP:jattiya-help} also introduced the definition of \emph{valency-helping} which unlike linearization-helping and universal-helping is defined on operation responses. Since correctness conditions like strict-linearizability only care about linearization order, it is more relevant to discuss help definitions that are defined using the linearization order of operations, and not using the response values. Nonetheless, for pedagogical purposes, we also introduce the definition of valency-helping (and mention its relevance in the context of some results).

\myparagraph{Valency-helping (\cite{DBLP:jattiya-help}, rephrased)}
Let $A$ be an implementation of object type $\tau$ whose inputs and output responses of operations belong to the set $\mathcal{V}$.
Let $C$ be a configuration of $A$ and $\pi$ an operation by process $p_i$. We say that $\pi$ is \emph{$v$-univalent} in $C$ if in every configuration $C'$ that is reachable from $C$ in which $\pi$ is complete, $\pi$ returns $v$; otherwise it is \emph{multivalent}.

Given an operation $\pi$ by process $p_i$, and a configuration $C$ such that $\pi$ is multivalent in $C$, we say that a process $p_j;j\neq i$ \emph{helps} process $p_i$ in $C$ if the next step of $p_j$ results a configuration $q(C)$ such that $\pi$ is $v$-univalent in $q(C)$ for some $v$. We say that $A$ has \emph{valency-helping} if it has a configuration $C$ such that some process $p_i$ helps process $p_j$ in $C$. Otherwise we say that $A$ is \emph{valency-help free} (Note that valency-helping is incomparable to both universal-helping and linearization-helping~\cite{DBLP:jattiya-help}).
%\end{proof}
%\subsection{Helping and Concurrency: Summary of Results}
%The relation between strict linearizabilty and the different forms of helping in the different model is presented in the following table.
%"No" indicates that there exists no algorithm which is strict-linearizable under the given model with the given helping.
%"Yes" indicates that there exists an algorithm which is strict-linearizable in the given model, and uses the helping as indicate by the table.
%

\section{Strict-linearizability vs. Linearization-helping}
\label{sec:SL-LH}

In this section we prove that strict-linearizability and linearization-helping are not bonds together.
That is, a strict-linearizable implementation in the crash-recovery model can either have or not have linearization-helping in the crash-stop model.

Claim~\tsref{cl:help1} proves that an implementation can be strict-linearizable in the crash-recovery model while being linearization-help free in the crash-stop model.
In Claim~\tsref{cl:DRAF} below we prove there in fact exist linearizable object type implementations that are wait-free with both linearization-helping and universal-helping, but are also strict-linearizable. We find this result to be somewhat surprising, since intuitively helping seems to contradict strict-linearizability -- an operation can be linearized by steps of other processes, thus if the owner of the operation crash the linearization point of the crashed operation may be after the crash. As we show, one can consider different linearization functions to proves the different properties of the implementation.

\subsection*{Double-Read-And-Fix Object Type}

\newcommand{\val}{\textsf{val}}
\newcommand{\readann}{\textsf{readAnn}}

The type used to prove the claim is a \texttt{Double-Read-And-Fix} (\draf) type object that supports \readop\ and \writeop\ operations. Its sequential specification is as follows: \writeop($v$) operation returns \ack. The first \readop\ operation returns an ordered pair $\langle v_1, v_2 \rangle$ of the last two preceding writes (if such exists, otherwise a unique $\bot$ symbol is used). Any later \readop\ returns $\langle v_1, v_2 \rangle$ as well. That is, the first \readop\ operation fix the value of the object for the rest of the history.

Algorithm \ref{alg:DRAF} presents a wait-free implementation of a \draf\ type using read, write and \CAS\ primitives. The variable \val\ is composed of four components. The first two, denoted $v_1,v_2$, represents the status of the object, i.e., the value of the last two writes (in order). The third component, named $counter$, is a global monotonically increasing counter, counting the number of writes performed to the object so far.
%That is, any write increases the value of the counter by 1.
The fourth component, named $gate$, is used to fix the status of the object after the first read operation. Once $gate$ is set to 1 no farther changes to \val\ are allowed.
In addition, a shared boolean variable \readann\ is used to announce a read operation started, so other processes can help it complete. A detailed description follows.

To perform a \writeop($u$) operation, process $p_i$ first read the counter stored in \val. Then it proceeds according to the value of \readann.
If \readann\ is set, i.e. its value is  1, implying some read operation has been invoked, $p_i$ helps it to complete by iteratively attempting to atomically set the $gate$ component of \val\ to 1 while not changing the other components, using a \CAS. Once $gate$ is set, $p_i$ returns \ack.
Otherwise no read operation has started. In such case, $p_i$ iteratively attempts, using \CAS, to atomically add its value $u$ to \val\ by setting the first component of \val\ to $u$, while increasing the counter by 1. $p_i$ stops and returns \ack\ once one of the following holds -- it successfully adds its value to \val\ in line \ref{DRAF-direct}; the counter was advanced by 2 or more since $p_i$ first read it; or, $gate$ has been set.

A \writeop\ operation is called \emph{direct-write} if it returns in line \ref{DRAF-direct}, \emph{indirect-write} if it returns in line \ref{DRAF-indirect}, and \emph{helping-write} if it returns in line \ref{DRAF-helping}. Notice that once a \writeop\ operation finds \readann\ to be 1 in line \ref{DRAF-if-gate-1} it guaranteed to be a helping-write.

\begin{algorithm}[]
        
        \footnotesize
		\begin{flushleft}
		    \textbf{Shared variables:} \\
		        $\qquad \val := \langle \bot, \bot, 0, 0 \rangle$ \\ 
		        $\qquad \readann := 0$ \\
		\end{flushleft}
    
        \begin{procedure}[H]
			\caption{() \small \writeop\ ($u$)}
			
			$localCounter := \val.counter$ \label{DRAF-counter-read} \;
			\uIf{$\readann = 1$ \label{DRAF-if-gate-1}}{
			    \While {true \label{DRAF-while-gate-set}} {
			        $\langle v_1, v_2, counter, cgate \rangle := \val$ \label{DRAF-helping-read} \;
			        \lIf{$cgate = 1$} {
			            \KwRet \ack \label{DRAF-helping}
			        }
			        $\CAS (\val, \langle v_1, v_2, counter, 0 \rangle, \langle v_1, v_2, counter, 1 \rangle)$ \label{DRAF-helping-CAS} \;
			    }
			}
			\uElse {
			    \While {true \label{DRAF-while-gate-not-set}} {
			        $\langle v_1, v_2, counter, cgate \rangle := \val$ \label{DRAF-indirect-read} \;
			        \uIf{$counter \geq localCounter+2$ or $cgate = 1$ \label{DRAF-indirect-if}}{
			            \KwRet \ack \label{DRAF-indirect}
			        }
			        \uIf{$\CAS (\val, \langle v_1, v_2, counter, 0 \rangle, \langle u, v_1, counter+1, 0 \rangle)$ \label{DRAF-direct-CAS}} {
			            \KwRet \ack \label{DRAF-direct}
			        }
			    }
			}
		\end{procedure}
		
		\begin{procedure}[H]
			\caption{() \small \readop\ ()}
			
			$\readann := 1$ \;
			\While {true \label{DRAF-while-read}} {
			        $\langle v_1, v_2, counter, cgate \rangle := \val$ \label{DRAF-read-read} \;
			        \lIf{$cgate = 1$} {\KwRet $\langle v_1, v_2 \rangle$ }
			        $\CAS (\val, \langle v_1, v_2, counter, 0 \rangle, \langle v_1, v_2, counter, 1 \rangle$ \;
			    }
		\end{procedure}
		
		\caption{\texttt{Double-Read-And-Fix}. Code for process $p_i$}
		\label{alg:DRAF}
    \end{algorithm}
    
\proofparagraph {Linearzability}
Given an execution $E$ we linearized it as follows.
A direct-write operation is linearized at the point where it performed its successful $CAS$ in line \ref{DRAF-direct-CAS}.
A helping-write, as well as an indirect-write which observes the $gate$ component is set to 1 (line \ref{DRAF-indirect-if}), is linearized on its last read of \val\ (lines \ref{DRAF-helping-read}, \ref{DRAF-indirect-read} respectively).
An indirect-write which observes the $counter$ component was increased by 2 (or more) in line \ref{DRAF-indirect-if} is linearized on its first step.
For \readop\ operations, let $\pi$ be the first \readop\ operation to set \readann\ to 1. Then $\pi$ is linearized at the point where the $gate$ component is set to 1 by some process. All other read operations are linearized at their last read of \val\ in line \ref{DRAF-read-read}.

The correctness of the above linearization order relies on the following simple observations:
(i) once \readann\ is set to 1 it remains so for the rest of the execution; (ii) any operation invoked after this point, either \readop\ or \writeop, must first make sure the $gate$ component in \val\ is set before returning. Therefore, no \writeop\ invoked after this point can change the values in \val; (iii) once the $gate$ component is set it remains so forever, since there is no \CAS\ operation with $gate=1$ as its old value.

Denote by $t$ the time when the $gate$ component was set to 1. Let $\pi_1,\pi_2$ be the last two direct-writes by the above linearization order, if such exists. It follows that both are linearized before time $t$.
Any \readop\ operation $\pi$ returns only after observing the $gate$ component is set, thus $\pi$ is linearized at time $t$ or later and returns the values of $\pi_1,\pi_2$, since after time $t$ the content of \val\ is fixed. This also implies the first \readop\ to set \readann\ must be active at time $t$, thus its linearization point is within its interval.

Finally, we prove that all other writes except for $\pi_1,\pi_2$ does not effect the response of other operations.
Any \writeop\ observing the $gate$ component is set is linearized after time $t$. Since this is not a direct-write, it does not write to \val\ (except for maybe setting $gate$), and thus it does not effect other operations response. A \writeop\ observing the counter was increased by at least 2 does not write to \val. However, since any direct-write increases the value of the counter by 1, this implies at least two direct-writes are linearized after $\pi$ was invoked. Hence, $\pi$ is linearized on its first step, and is overridden by later writes.

\proofparagraph{Wait-Freedom}
A direct-write and indirect-write are wait-free -- failed \CAS\ in line \ref{DRAF-direct-CAS} implies the content of \val\ has been changed. However, any write to \val\ either set the $gate$ to 1, or increases $counter$ by 1. Therefore, the while loop in line \ref{DRAF-while-gate-not-set} can be iterated at most twice.

A \readop, as well as helping-write, are wait-free for the following argument. Notice that once the $gate$ component is set, any operation completes in a finite number of steps. Therefore, it suffice to prove the $gate$ is set after a finite number of steps. Consider some process $p_i$ attempting to set the $gate$, i.e., either a \readop\ or a helping-write operation. Any failed \CAS\ attempt implies the value of \val\ has been changed. If the $gate$ was set, then we are done. Otherwise, some other process $p_j$ performed a successful direct-write. Hence, its current operation is complete. However, on its next operation, either \readop\ or \writeop, it will observe \readann\ is set, thus it will try to set the $gate$ component as well. As a result, after at most $N$ failed $CAS$ attempts, all processes are either not active, or trying to set the gate, and at least one must succeed. 

\proofparagraph{Linearization-Helping}
Since \draf\ is an exact-order type, by Theorem 4.18 in \cite{help15}, any wait-free implementation of an exact-order type using read-write registers and \CAS\ has also linearization-helping.

\proofparagraph{Universal-Helping}
Denote by $f$ the linearization function used to prove linearizability. We next prove that a similar linearization function proves universal-helping.
Let $E$ be an execution, and $\pi$ be an operation in it by some process $p_i$. To prove universal-helping we prove that after a finite number of operations by any other process $\pi$ has a linearization point. We consider two cases based on the type of operation.

$\pi$ is a \readop\ operation. Any other process $p_j$ completing two operations after $\pi$ writes 1 to \readann\ must set the $gate$ component to 1, unless it was already set. Although $f$ linearize \readop\ operations on their last access to \val, we can change it such that once the $gate$ is set, all pending \readop\ operations are linearized (in an arbitrary order), while any future \readop\ is linearized on its first step. The correctness proof for the new linearization-order is similar to the one given for $f$.
This implies that once the $gate$ is set, $\pi$ have a linearization point, and this linearization point remains the same for any extension.

$\pi$ is a \writeop\ operation. Consider a different process $p_j$ completing two operations after $\pi$ reads the counter in line \ref{DRAF-counter-read}. Denote by $t'$ the time when $p_j$ complete the second operation.
Then in any of $p_j$ operations, it either observes $gate$ was set, or $counter$ was incremented (by $p_j$ or some other process). Hence, at $t'$ either $gate$ is set or that $counter$ was incremented by 2 or more.
We consider three cases:
(1) at $t'$ $\pi$ already written its value to \val, i.e., it is a direct-write. Thus by $f$ it has a linearization point before $t'$;
(2) at $t'$ the $gate$ is set. Similar to the read case, we can change $f$ such that all non direct-write that are active at the point where the $gate$ is set are linearized at this point (in an arbitrary order, but after the reads). Any later \writeop\ is linearized on its first step;
(3) at $t'$ the $counter$ was incremented by two or more since $\pi$ first read it. Then $\pi$ is linearized on its first step by $f$, and in particular before $t'$;
In all cases we get that $\pi$ has a linearization point in $t'$, and it remains the same for any extension.

\proofparagraph{Strict-Linearizability}
Consider some execution $E$ in the individual crash-recovery model. Assume $p_i$ crashed while executing an operation $\pi$.
Assume $\pi$ is a \writeop\ operation.
If $\pi$ is a direct-write, meaning that it has written to \val\ in line \ref{DRAF-direct-CAS} before the crash, then $\pi$ has a linearization point before the crash. In any other case, i.e., indirect-write and helping-write, $\pi$ does not writes its value to \val, thus we can consider it as having no linearization point at all.

Assume now $\pi$ is a \readop\ operation. If $\pi$ was invoked after $gate$ component was set, then we can consider it has having no linearization point, since it does not effect any other operation. The same can be done in case $gate$ was not set in $E$.
Otherwise, $\pi$ was invoked before $gate$ was set, and the $gate$ is set in $E$. Denote by $t$ the time when $gate$ was set. If there is an active non-crashed \readop\ operation at time $t$, we can linearize it at $t$, and all crashed \readop\ operations can be removed from the history without effecting any other operation. Therefore, we are left with the case where all read operations are crashed before time $t$.

w.l.o.g. assume $\pi$ is the first operation to set \readann\ and it is the only \readop\ operation invoked before time $t$. As we next prove, linearizing $\pi$ at $t$ does not violates strict-linearizability. Once this is done, all other crashed \readop\ operations, if any, can be removed from the history, and strict-linearizability follow.
Notice that $t$ is after the crash point of $p_i$, implying some other operation sets the $gate$. This seems to contradict strict-linearizability. Nevertheless, we next prove any operation that is linearized before $t$, is either pending or completed when $p_i$ crashed. Therefore, all those operations, including $\pi$ itself, can be seen as having linearization points before the crash, and this concludes the proof.

Let $\sigma$ be an operation by some process $p_j$ invoked after the crash of $\pi$.
If $\sigma$ is a \readop\ operation, then by our simplifying assumption $\pi$ is the only \readop\ before time $t$, thus $\sigma$ was invoked after time $t$, and in particular is linearized after time $t$.
Otherwise $\sigma$ is a \writeop\ operation. Notice that it can be invoked after the crash of $p_i$ and before time $t$. However, since $\pi$ already set \readann\ to 1, $\pi_j$ must be a helping-write. Therefore it completes only after the $gate$ is set, that is after time $t$, and it is also linearized after time $t$.
This proves that any operation that is linearized before time $t$ is invoked before the crash of $p_i$.

\begin{claim}
\label{cl:DRAF}
There exists an implementation $A$ of an object type $\tau$ such that $A$ is linearizable, wait-free and has both linearization-helping and universal-helping in the crash-stop model. Moreover, $A$ is strict-linearizable in the individual crash-recovery model.
\end{claim}
\begin{lemma}
\label{lm:SL-vs-LH}
    An implementation being strict-linearizable in the crash-recovery model is independent of it satisfying linearization-helping in the crash-stop model.
\end{lemma}

\begin{proof}
    Follows directly from Claims \tsref{cl:help1} and \tsref{cl:DRAF}.
\end{proof}
\section{Help-freedom vs. Strict-linearizability}
\label{sec:HelpFree-SL}

In this section we prove an implementation can have no helping (both linearization and universal helping) while still being not strict-linearizable. The result is counter intuitive, since no linearization-helping seems to imply only the owner of an operation $p$ can cause it to be linearized by its own step. Hence, in case of a crash, either the pending operation was already linearized by a step of $p$, or that it is yet to be linearized and $p$ takes no more steps after the crash, thus the operation will have no linearization point. As we prove, forcing a linearization function to have no decided-before relation between two operations, even after the operations are linearized and completed, allows us to derive such a counter-example.

We then preclude such behaviours by posing a condition on the linearization function. In a nutshell, we consider only functions such that after two operations are linearized, in any extending execution the linearization function must linearized both in the same order. We note that for the best of our knowledge, any known linearizable implementation has such a linearization function.
%this restriction is compatible with linearizability, since an implementation is linearizable if and only if such a linearization function exists.
We prove that under this restriction, linearization-help free indeed implies strict-linearizability.

\subsection{Sticky-Bit Object}
\label{sec:eq2}

\newcommand{\vall}{\textsf{val}}
\newcommand{\ann}{\textsf{ann}}

A \sticky\ object type $\tau$ is the most simple form of multi-reader multi-writer register. Its value is initially 0, and it supports \setop\ and \readop\ operations. A \setop\ operations returns $ack$, while a \readop\ operation returns 0 if there is no \setop\ preceding it in the history, and 1 otherwise.
Algorithm \tsref{alg-Sbit} presents a \sticky\ implementation $I$ using single bit registers, such that $I$ is linearizable, wait-free, linearization-help free, and universal-help free in the crash-stop model. However, as we prove, $I$ is not strict-linearizable in the system-wide crash-recovery model.

\setop\ operation simply writes 1 to \vall. However, $p_0$ and $p_1$ executes \setop\ in a different manner, by first announcing their operation by setting a bit in the \ann\ array, and only then writing 1 to \vall.
Two different \readop\ implementations are provided. For clarity, we refer to it in the code as \readopA\ and \readopB. We assume some processes use only \readopA, and all others uses only \readopB. As we prove, the claim holds as long as there is at least one process different then $p_0,p_1$ using \readopA, and at least one such process using \readopB. As such, we assume it holds, and do not specify the exact set of processes using each of the different \readop\ implementations.
Wait-freedom follows directly from the code.

\begin{algorithm}[]
        %\nonl
        %\removelatexerror
        \footnotesize
        
		\begin{flushleft}
		    \textbf{Shared variables}:\\
		        \hspace{6 mm} \vall\ - read/write register, initially 0 \\
    		    \hspace{6 mm} $\ann[2]$ - boolean array, initially \False
		\end{flushleft}
    
    \begin{multicols*}{2}
        
        \begin{procedure}[H]
			\caption{() \small $\setop\ ()$ \newline
			    \texttt{code for process $p_i \notin \{p_0, p_1\}$}}
			
			$\vall := 1$ \;
			\KwRet \ack \;
		\end{procedure}
		
		\begin{procedure}[H]
			\caption{() \small $\setop\ ()$ \newline
			    \texttt{code for process $p_i \in \{p_0, p_1\}$}}
			
			$\ann[i] := \True$ \;
			$\vall := 1$ \;
			\KwRet ack \;
		\end{procedure}
		
	\columnbreak
		
		\begin{procedure}[H]
			\caption{() \small $\readopA\ ()$ \newline
			    \texttt{code for process $p_i$}}
			
			$res := \vall$ \;
			\KwRet res \;
		\end{procedure}
		
		\begin{procedure}[H]
			\caption{() \small $\readopB\ ()$ \newline
			    \texttt{code for process $p_i$}}
			
			\uIf {$\ann[0] = true$ and $\ann[1] = true$} {
			    $\vall := 1$ \;
			}
			$res := \vall$ \;
			\KwRet res \;
		\end{procedure}
		
	\end{multicols*}
		
	\caption{\sticky\ Object}
	\label{alg-Sbit}
\end{algorithm}

%\begin{proof}

\proofparagraph{Linearizability}
A \readop\ operation is linearized at the point where it reads \vall. For \setop\ operation we consider two different cases, based on the first write to \vall\ by some process $p_i$ performing an operation $\pi$.

Case I -- $\pi$ is a \setop\ operation. Then, we linearize any \setop\ operation on its write to \vall. As a result, any other \setop\ is linearized after $\pi$, a \readop\ returning 0 is linearized before $\pi$, while any \readop\ returning 1 is linearized after $\pi$.

Case II -- $\pi$ is a \readop\ operation, namely, it is \readopB. In such case, process $p_i$ observed both $\ann[0]$ and $\ann[1]$ are true, that is both $p_0$ and $p_1$ already invoked a \setop\ operation. Denote these operations by $\setop_0$, $\setop_1$ respectively. Since $p_i$ is the first to write to \vall, no \readop\ operation returned 1 before $p_i$ writes to \vall. We linearized $\setop_0$ at the write to \vall\ by $\pi$, while any other \setop\ is linearized on the step where it writes to \vall.
As in case I, a \readop\ operation reading \vall\ before it was first written is linearized before $\setop_0$ and returns 0, while any other \readop\ operation is linearized after this point and returns 1. This concludes the linearizability proof.

\proofparagraph{Universal-help Free}
Consider an execution $E$ in which $p_0$ invokes a \setop\ operation, writes 1 to $\ann[0]$ and halts. Follow that process $p \neq p_0$ performs an infinite sequence of \readopA\ operations. Then, all \readop\ operations return 0, thus the \setop\ operation is pending and has no linearization point during the entire execution. Hence Algorithm \tsref{alg-Sbit} is universal-help free.

\proofparagraph{Linearization-help Free}
Consider some execution $E$ of the implementation. Following the linearization points defined above, in case I any operation is linearized on a step by the process performing it, and thus $E$ is linearization-help free. Nevertheless, in case II, a step by process $p_i$ performing \readopB\ operation results a linearization point for a \setop\ operation by $p_0$, while any other operation is linearized on a step by the process performing it. This seems to imply linearization-helping. However, the same argument holds for a linearization order in which we choose to linearize $\setop_1$ instead of $\setop_0$ at the first write to \vall. Notice that changing the linearization order in this manner does not effect any other operation. Therefore, we can choose a linearization function $f$ such that if case II holds, in some extensions $\setop_0$ is linearized first, while in others $\setop_1$. For example, based on the number of operations in the execution (odd or even).

Under the linearization function $f$, $E$ has both an extension where $\setop_0$ is linearize before $\setop_1$, and vice versa. Hence, by definition $\setop_0$ is not decided before $\setop_1$ at any point. Moreover, according to $f$, before $\setop_0$ writes to \vall\ it can be linearized before any other operation that is yet to read or write \vall. Once $p_0$ writes to \vall\ it determines the last possible linearization point (in any extension) for $\setop_0$. In other words, only a step by $p_0$ can cause $\setop_0$ to be decided before some other operation. The same argument applies to $\setop_1$. Thus, function $f$ proves $I$ is linearization-help free.

\proofparagraph{Strict-Linearizability}
Consider the following execution $E$ of the implementation $I$ in the crash-recovery model. Each of $p_0, p_1$ invokes a \setop\ operation and writes to \ann, followed by a system-wide crash. Upon recovery, some process $p_i$ invokes and complete a \readopA\ operation, thus returning 0. Follow that, some other process $p_j$ performs \readopB\ operation to completion and returns 1. It follows that in any linearization order for $E$ one of the \setop\ operations of either $p_0$ or $p_1$ (or both) must be linearized in-between the two \readop\ operations of $p_i$ and $p_j$, and in particular, after the crash point. This contradicts strict-linearizability.
%\end{proof}
%
\begin{claim}
\label{cl:sticky}
There exists a wait-free linearizable implementation $I$ of an object type $\tau$ such that $I$ is linearization-help free and universal-help free in the crash-stop model, while $I$ is not strict-linearizable in the system-wide crash-recovery model.    
\end{claim}

\begin{remark}
    Implementation $I$ being not strict-linearizable in the system-wide new identifiers crash-recovery model imply $I$ is also not strict-linearizable in all other crash-recovery models - individual (system-wide) new (old) identifiers crash-recovery model.
    Therefore, Lemma \tsref{lm:sticky} holds for all crash-recovery models.
\end{remark}

\begin{lemma}
    Strict-linearizability and linearization-helping are independent
\end{lemma}

\begin{proofsketch}
        Lemma~\ref{lm:SL-vs-LH} proves an implementation can be strict-linearizable in the crash-recovery model while having or not linearization-helping in the crash-stop model. For the not strict-linearizable case, Claim~\ref{cl:sticky} proves an implementation can be not strict-linearizable in the crash-recovery model, while being linearization-help free in the crash-stop model. The last case remains is an implementation being not strict-linearizable in the crash-recovery model, while satisfy linearization-helping in the crash-stop model. This case is almost trivial, and can be proven using many known implementations with linearization-helping. For example, in the Binary Search Tree implementation of Ellen et al.~\cite{EllenFRB10}, an update operation may mark a node, and later in the execution the operation can be completed by a different process. To prove the implementation is not strict-linearizable, consider the scenario where the process crash after the marking.
\end{proofsketch}

\subsection{An Equivalence between Linearizability and Strict-linearizability}
The formalization of linearization-helping specifies a linearization function $f$ such that for any history $H$ it produce its linearization order $f(H)$. 
%However, except for linearization no other restriction are required from $f$. 
As the sticky-bit implementation demonstrates, the formalization of linearization-helping leads to executions where the decided-before order is not well-defined, and this may break the concept of helping leading to non-intuitive results.
Specifically, $f$ can linearize operations of some history $H$ in different order for different extensions of $H$.
For example, consider a stack implementation, and a history $H$ where two \emph{pop} operations $\pi_1, \pi_2$ are executed concurrently to completion starting from the initial configuration. Both operations returns empty, and can be linearized in any order. Therefore, in different extensions of $H$, $f$ may sometimes linearize $\pi_1$ before $\pi_2$ and sometimes vice versa. Thus, although both operations completed, there is no decided-before order between the two.

In this section we restrict the discussion to a more natural \emph{prefix-respecting} linearization function precluding such a behaviour. We emphasize that to the best of our knowledge, any known implementation has such a linearization function.

\begin{definition}
We say that a linearization function $f$ is \emph{prefix-respecting} if for any execution $E$ and an execution $F$ extending it, $<_{f(E)} \subseteq <_{f(F)}$. In other words, for any two operations $\pi_1,\pi_2 \in E$, if $\pi_1 <_{f(E)} \pi_2$ then $\pi_1 <_{f(F)} \pi_2$.
%$f(F)$ agrees on the order of operations linearized in $f(E)$.
%the following holds: for any two operations $\pi_1, \pi_2$, if $\pi_1$ precedes $\pi_2$ in $f(E)$ (and in particular, both appears in $f(E)$), then $\pi_1$ precedes $\pi_2$ in $f(F)$.
\end{definition}
Strong-linearizability \cite{golab-strong} requires the linearization order of an execution $E$ to be a prefix of the linearization order of any extension of $E$. Although this seems to bear similarity to prefix-respecting, strong-linearizability is a more restrict requirement.
Given an execution $E$, strong-linearizability requires that if an operation $\pi$ is linearized in $f(E)$, then in any extending execution no operation $\sigma$ that is pending in $E$ can be linearized before $\pi$, even if $\sigma$ was invoked before $\pi$ and it is not linearized in $f(E)$. On the other hand, prefix-respecting allows such a scenario, as long as you do not change the linearization order of operations that are already linearized in $f(E)$.

A key feature of prefix-respecting linearization function is that given an execution $E$, if some operation $\pi_1$ is linearized before $\pi_2$ in $f(E)$, then it holds that $\pi_1$ is decided before $\pi_2$ in $E$ (with respect to $f$). Hence, and by abuse of notation, we may say $\pi_1$ is linearized before $\pi_2$ while referring to the decided before order.
Restricting our discussion to linearization-helping based on prefix-respecting linearization functions only, we prove that a linearization-help free implementation $A$ of an object type $\tau$ in the crash-stop model is also strict-linearizable in the crash-recovery model. 
Roughly speaking, under this restriction, if $A$ is linearization-help free and some operation $\pi$ crashes, then either it has a linearization point before the crash, or that it has no linearization point in any extending execution, as steps by other processes can not cause it to be linearized.

%The proof applies to a set of objects for which each process can apply any operation at any point of the abstract history. An example for such an object is snapshot, stack and queue. To the best of our knowledge, any known long-lived object satisfy this condition. An example for an object that does not satisfy the condition can be drawn by restricting known object. For example, consider a stack where process $p$ is allowed to perform $Push(i)$ only after completing $Push(i-1)$. In such object, a process is not allowed to invoke $Push(2)$ as its first operation.

%\ohad{move total objects definition to the model section?}
%The proof applies to \emph{total} object types and implementations. 
%
\begin{figure}
    \centering
    \includegraphics[width=\textwidth]{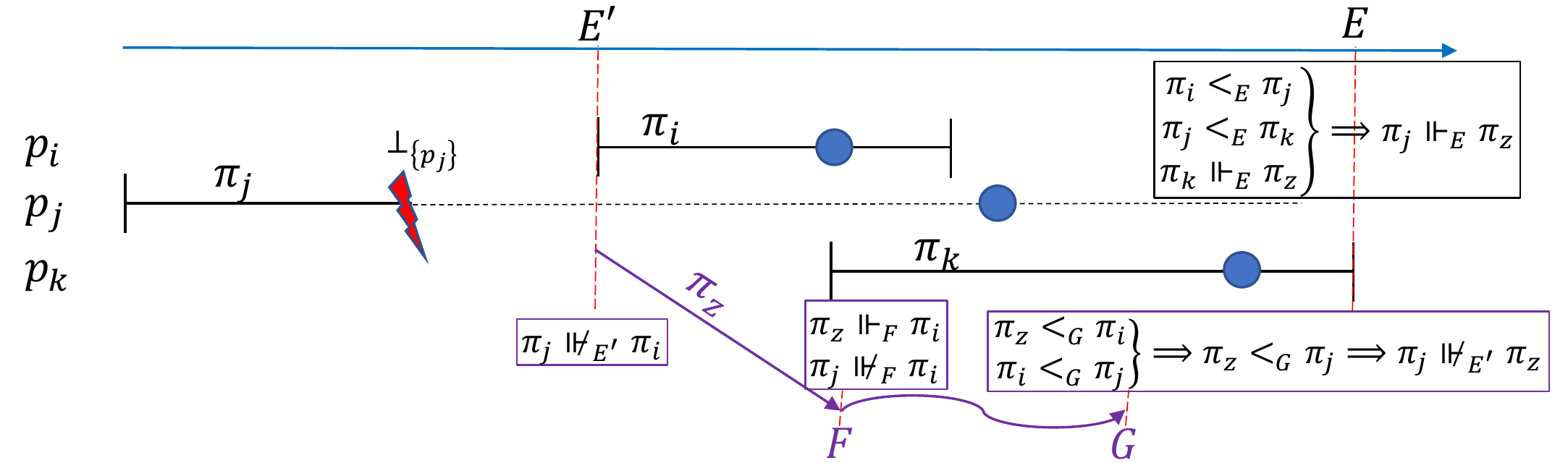}
    \caption{\small Illustration of Lemma \tsref{lm:prefix-strict} proof. Execution $E$ is illustrated by black lines, where blue dots are linearization points according to $f$. Purple elements describes the construction in the proof. Each frame contains the properties of the respective resulted execution.}
    \label{fig:prefix}
\end{figure}

\begin{lemma}
    \label{lm:prefix-strict}
    Let $A$ be an obstruction-free implementation of a total object type $\tau$ such that $A$ is not strict-linearizable in the individual crash-recovery model.
    Then, any prefix-respecting linearization function $f$ of $A$ imply linearization-helping in the crash-stop model.
\end{lemma}

\begin{proof}
        
    In high-level,
    since $A$ is not strict-linearizable there exists an execution $E$ of $A$ contradicting strict-linearizability.
    Assume towards a contradiction there exists a prefix-respecting linearization function $f$ such that $A$ is linearization-help free with respect to $f$.
    $E$ does not satisfy strict-linearizability, thus there exists a crashed operation $\pi_j$ which is linearized after the crash of its owner process $p_j$ in $f(E)$. Consider some operation $\pi_z$ by a different process $p_z$. We prove that $\pi_j$ is not decided before $\pi_z$ when $p_j$ crash, while $\pi_j$ is decided before $\pi_z$ in $E$. However, after the crash of $p_j$ all steps are by processes different then $p_j$, thus contradicting linearization-help free.
    A detailed proof follows.
    
    %\ohad{the following paragraph should be in the model / beginning, elaborating why we can consider the same implementation and "same" executions in both models}
    %Moreover, since the model is atomic, there is no data lost in case of a crash.
    
    \proofparagraph{Notation and Observations}
    For the rest of the proof we fix the linearization function $f$. Therefore, all claims are with respect to $f$, and we do not mention it. In addition, by abuse of notation we write $\pi_1 <_F \pi_2$ to denote $\pi_1 <_{f(F)} \pi_2$.
    Given two operations $\pi_1,\pi_2$ we write $\pi_1 \Vdash_{F} \pi_2$ ($\pi_1 \nVdash_{F} \pi_2$) if $\pi_1$ is decided before (not decided before) $\pi_2$ in the execution $F$ with respect to $f$.
    We say that a process $p_i$ is \emph{idle} in execution $F$ if $p_i$ has no pending operation in $F$.
    Let $F$ be an execution, and $\pi_1, \pi_2$ be operations by $p_1, p_2$ respectively. The following claims hold:
    \addtolength\leftmargini{0.75em}
    \begin{enumerate}[label=\textbf{P\arabic*:}, ,ref=P\arabic*]
        \item   \label{prop:linearize-imply-decided}
                $\pi_1 <_F \pi_2$ imply $\pi_1 \Vdash_F \pi_2$. \\
               {\small Explanation: $f$ is prefix-respecting, thus lineatization order implies decided before order. \par}
        \item  \label{prop:not-decided-to-linearized}
                $\pi_1 \nVdash_F \pi_2$ implies there exists an extension of $F$ such that $\pi_2$ is linearized before $\pi_1$. \\
                {\small Explanation: follows immediately from the definition of decided before order. \par}
        \item   \label{prop:decided-is-transitive} 
                $\Vdash$ is transitive: $\pi_1 \Vdash_E \pi_2 \wedge \pi_2 \Vdash_E \pi_3 \Longrightarrow \pi_1 \Vdash \pi_3$.
                %{\small Explanation: if $\pi_1$ is decided before $\pi_2$, and $\pi_2$ is decided before $\pi_3$, then $\pi_1$ is decided before $\pi_3$. \par}
        \item   \label{prop:p-free-extension}
                $\pi_1 \nVdash_F \pi_2$, implies $\pi_1 \nVdash_{F \cdot F'} \pi_2$ for any $p_1$-free extension $F'$ of $F$. \\
                {\small Explanation: since $A$ is linearization-help free with respect to $f$, steps by processes different then $p_1$ can not cause $\pi_1$ to be decided before $\pi_2$. \par}
    \end{enumerate}
    
    \proofparagraph{Detailed Proof}
    $f(E)$ is a linearization of $E$,  hence a contradiction to strict-linearizability is due to a crashed operation $\pi_j$ by some process $p_j$. That is, $\pi_j$ is linearized after the crash of $p_j$ in $f(E)$. Next we prove this implies $\pi_i <_E \pi_j <_E \pi_k$ for some operations $\pi_i,\pi_k$, by processes $p_i,p_k$ respectively, such that $\pi_i$ was invoked after the crash of $p_j$.
    
    Assume towards a contradiction one of $\pi_i, \pi_k$ does not exists.
    If no such $\pi_i$ exists, then any operation that is linearized before $\pi_j$ in $f(E)$ was invoked before the crash of $p_j$. As a result, all those operations, including $\pi_j$, have a linearization point before the crash of $p_j$.
    If no such $\pi_k$ exists, then no operation is linearized after $\pi_j$ in $f(E)$. We consider a linearization order similar to $f(E)$ in which $\pi_j$ have no linearization point.
    We get a linearization order for $E$ in which $\pi_j$ is either linearized before the crash of $p_j$, or have no linearization point at all, contradicting the fact that $\pi_j$ must be linearized after its crash.
    
    Denote the execution up to the invocation of $\pi_i$ by $E'$.
    Let $p_z$ be an idle process in $E'$, and let $\pi_z$ be an operation by $p_z$. We prove (1) $\pi_j \nVdash_{E'} \pi_z$; (2) $\pi_j \Vdash_{E} \pi_z$. However, the extension leading from $E'$ to $E$ is $p_j$-free, since $p_j$ crashed in $E'$, contradicting \ref{prop:p-free-extension}.
    Figure \ref{fig:prefix} depicts the structure of the execution constructed by our proof.
    
    For simplicity, we first prove the lemma under the following two assumptions: (i) $\pi_i,\pi_k$ are completed in $E$; (ii) starting from $E'$ only $p_i$ and $p_k$ takes steps in $E$, while all other processes are idle.
    To prove (1) we construct an execution starting from $E'$ such that $\pi_z$ is linearized before $\pi_j$.
    Consider a solo run of $p_z$ starting from $E'$ such that $\pi_z$ is completed and linearized. Denote the resulted execution by $F$. Since $\pi_i$ is yet to be invoked in $F$ we get $\pi_z \Vdash_{F} \pi_i$.
    Notice that $\pi_j \nVdash_{E'} \pi_i$, since there is an extension of $E'$ leading to $E$ in which $\pi_i <_E \pi_j$. The extension leading from $E'$ to $F$ is $p_i$-free, thus by \ref{prop:p-free-extension} $\pi_j \nVdash_{F} \pi_i$. By \ref{prop:not-decided-to-linearized} there exists an extension of $F$ to an execution $G$ such that $\pi_i <_G \pi_j$. Since $\pi_z \Vdash_{F} \pi_i$, in any extension of $F$, and in particular in $G$, $\pi_z$ is linearized before $\pi_i$, which is linearized before $\pi_j$. %This completes the proof of (1).
    We next prove (2).
    Since $\pi_j <_E \pi_k$, by \ref{prop:linearize-imply-decided} $\pi_j \Vdash_E \pi_k$. By our assumption $\pi_k$ is completed in $E$ before $\pi_z$ is invoked, thus $\pi_k \Vdash_E \pi_z$. Since $\Vdash$ is transitive, $\pi_j \Vdash_E \pi_z$.
    
    Next, we prove the general case without the simplifying assumptions made above. To do so, we construct an execution in which the two assumptions holds. That is, after $E'$ only $p_i$ and $p_k$ may take steps while all other processes are idle, and both $\pi_i, \pi_k$ are completed. The construction proceeds in phases.
    
    \proofparagraph{Phase 1}
    Denote the execution up to the crash point of $p_j$ by $E_0$. Starting from $E_0$, as long as there is a pending operation $\pi$ by a process $p \notin \{p_i,p_j,p_k\}$, we let $p$ perform a solo run and complete its pending operation.
    Denote the resulted execution by $E_1$. By the construction, in $E_1$ all processes except maybe for $p_i,p_j,p_k$ are idle.
    Since there is an extension of $E_0$ such that $\pi_i <_E \pi_j <_E \pi_k$, we get $\pi_j \nVdash_{E_0} \pi_i$ and $\pi_k \nVdash_{E_0} \pi_j$. However, all steps in the extension of $E_0$ to $E_1$ are by processes different then $p_i,p_j,p_k$, thus by \ref{prop:p-free-extension} $\pi_j \nVdash_{E_1} \pi_i$ and $\pi_k \nVdash_{E_1} \pi_j$.
    
    \proofparagraph{Phase 2}
    Starting from $E_1$, let $p_i$ complete its pending operation, if such exists, followed by an execution of $\pi_i$ to completion. Denote the resulted execution by $E_2$.
    Since all steps in the extension from $E_1$ to $E_2$ are by $p_i$, following \ref{prop:p-free-extension} we have $\pi_j \nVdash_{E_2} \pi_i$ and $\pi_k \nVdash_{E_2} \pi_j$.
    We now prove that since $\pi_i$ is completed in $E_2$ and linearized, this implies $\pi_i \Vdash_{E_2} \pi_j$.
    
    Assume towards a contradiction $\pi_i \nVdash_{E_2} \pi_j$. Since $\pi_j \nVdash_{E_2} \pi_i$ holds, by \ref{prop:not-decided-to-linearized} there exists an extension of $E_2$ in which $\pi_i$ is linearized and decided before $\pi_j$. However, $\pi_i$ is completed and have no more steps to take. Thus, there is a step by an operation different then $\pi_i$ causing $\pi_i$ to be decided before $\pi_j$. This contradicts the fact that $f$ is linearization-help free.
    
    \proofparagraph{Phase 3}
    Finally, starting from $E_2$ we let $p_k$ complete any pending operation different then $\pi_k$, if such exists, followed by an execution of $\pi_k$ to completion in a solo-run manner. Denote the resulted execution by $E_3$.
    We prove $E_3$ satisfies all the required conditions. First, notice that $p_j$ crashed in $E_3$. In addition, starting from the invocation step of $\pi_i$, only $p_i$ and $p_j$ takes steps, while any other process is idle, and both $\pi_i,\pi_k$ are completed in $E_3$. It remains to prove $\pi_i <_{E_3} \pi_j <_{E_3} \pi_k$.
    Since $E_3$ is an extension of $E_2$ we have $\pi_i \Vdash_{E_3} \pi_j$. As we next prove, $\pi_j \Vdash_{E_3} \pi_k$. Moreover, by the construction $\pi_i <_{E_3} \pi_k$. 
    This proves $\pi_j$ must be linearized between $\pi_i$ and $\pi_k$.  This completes the proof.
    
    To prove $\pi_j \Vdash_{E_3} \pi_k$, we prove it holds in $E_0$, and thus in any extension of $E_0$.
    Assume towards a contradiction $\pi_j \nVdash_{E_0} \pi_k$. Since $p_j$ crashed in $E_0$, the extension leading to $E$ is $p_j$-free. By \ref{prop:p-free-extension} $\pi_j \nVdash_{E} \pi_k$, contradicting the fact that $\pi_j <_E \pi_k$.

\end{proof}

\begin{corollary}
\label{lm:prefix}
    Let $A$ be an obstruction-free implementation of a total object type $\tau$ such that $A$ is linearizable and linearization-help free in the crash-stop model. Then $A$ is strict-linearizable in the individual crash-recovery model.
\end{corollary}

%\section{Does Strict-linearizability imply Help-freedom?}
\section{Strict-Linearizability vs. Universal-helping}
\label{sec:SL-UH}
%
%This section consider the new identifier crash-recovery model.

%\subsection{Universal-Helping vs. Linearization-Based Helping}
%Universal helping and linearization-based helping are incomparable. 
%\subsection{Strict-Linearizability vs. Universal-helping}

In this section we prove that strict-linearizability and universal-helping are independent. We then prove that for a large class of object types, any strict-linearizable implementation in the crash-recovery model is universal-help free in the crash-stop model.

\begin{lemma}
    \label{lm:SL-vs-UH}
    Strict-linearizability and universal-helping are independent
\end{lemma}

\begin{proofsketch}
    We prove an implementation being strict-linearizable or not in the crash-recovery model is independent of it having or not having universal-helping in the crash-stop model. We do so by proving that for any of the four possible combinations there exists an implementation.

    \proofparagraph{Strict-linearizability + Universal-helping}
    Follows from Claim~\ref{cl:help1} and \ref{cl:DRAF}.
    
    \proofparagraph{Strict-linearizability + Universal-help free}
    This direction is intuitive, since naturally strict-linearizability seems to contradict universal-helping. Indeed, any strict-linearizable implementation we are familiar with is universal-help free in the crash-stop model. For example, consider the well-known Harris linked-list \cite{Harris01}. It is strict-linearizable in the crash-recovery model since any operation is linearized on a step by its owner process. However, it is universal-help free in the crash-stop model since processes do not help each other (except for physical removal of nodes), thus an insert operation will never be completed if its owner halts before adding the key to the list.
    
    \proofparagraph{not Strict-linearizable + Universal-helping}
    This direction is intuitive as well, since in most cases universal-helping contradicts strict-linearizability. Consider Herlihy  universal-construction \cite{Herlihy91} applied to a stack object type. It has universal-helping in the crash-stop model. However, it is not strict-linearizable in the crash-recovery model, since a process may crash after announcing a pop operation, and later that operation can be completed by a different process. This may result a return value of a push operation which was invoked after the crash. This implies the linearization point of the pop operation must be after the crash.

    \proofparagraph{not Strict-linearizable + Universal-help free}
    Follows from Claim~\ref{cl:sticky}
    
\end{proofsketch}

\subsection{Equivalence between Strict-linearizability and Universal-help Freedom}

\proofparagraph{Notation}
Given a sequential history $H$ and an operation $\pi$, we denote by \emph{$H+\pi$} the set of all sequential histories obtained by including $\pi$ (i.e., its invocation and response) in $H$. The same is defined for more than one operation.
Two sequential histories $H_1, H_2$ are \emph{distinct} if there exists an operation $\pi$ in both such that its response is different. If $\pi$ is in some (common) sub-history $H$, we say that $H_1, H_2$ are \emph{distinct for an operation in $H$}.

\begin{definition}[\textbf{Order-dependent type}]
\label{def:order-dependent}
An object type $\tau$ is \emph{order-dependent} if there exists an infinite sequential history $H$ and two operations $\pi_1, \pi_2$ such that the following holds:
\addtolength\leftmargini{1.5em}
\begin{enumerate}[label=\textbf{OD\arabic*:} ,ref=OD\arabic*]
    \item \label{def:od-adding-distinct}
    For any histories $H_1 \in (H+\pi_1) \cup (H+\pi_2)$,  $H_2 \in H+\pi_1+\pi_2$, any two histories in $\{H,H_1,H_2\}$ are distinct for some operation in $H$.
    %Notice: different pairs may be distinct for different operations in $H$.
    \item \label{def:od-order-distinct}
    $\pi_1 \cdot \pi_2 \cdot H$ and $\pi_2 \cdot \pi_1 \cdot H$ are distinct for some operation in $H$.
\end{enumerate}
\end{definition}
%
%\textcolor{blue}{Make clear references to claims and lemmas in appendix}
Order-dependant types include many known objects, such as queue and stack. In a nutshell, an order-dependent type have two operations such that adding exactly one of them or both changes the response of some other operation. Moreover, the order in which both operations are performed (starting from the initial configuration) effects the response of some other operation.

%Additionally, Lemma~\tsref{lm:exact-lemma} proves the existence of an exact-order type wait-free linearizable implementation using only read/write registers and $CAS$ that has both linearization-helping and universal-helping in the crash-stop model, but is also strict-linearizable in the individual crash-recovery model.

\begin{claim-subsection}
    A queue is an order-dependent type.
\end{claim-subsection}

\begin{proof}
    Let $\pi_1 = Enqueue(1)$, $\pi_2 = Enqueue(2)$, and $H$ be an infinite sequence of $Dequeue$ operations.
    Assuming the queue is initially empty, all operations in $H$ returns \textit{empty}. Adding only one of $\pi_1, \pi_2$ cause a single operation in $H$ to return a value different then \textit{empty}, while adding both effects the response of two operations in $H$.
    In addition, the first $Dequeue$ operation in $H$ returns different response in the histories $\pi_1 \cdot \pi_2 \cdot H$ and $\pi_2 \cdot \pi_1 \cdot H$.
\end{proof}

\begin{claim-subsection}
    A stack is an order-dependent type.
\end{claim-subsection}

\begin{proof}
    Let $\pi_1 = Push(1)$, $\pi_2 = Push(2)$, and $H$ be an infinite sequence of $Pop$ operations.
    Assuming the stack is initially empty, all operations in $H$ returns \textit{empty}. Adding only one of $\pi_1, \pi_2$ cause a single operation in $H$ to return a value different then \textit{empty}, while adding both effects the response of two operations in $H$.
    In addition, the first $Pop$ operation in $H$ returns different response in the histories $\pi_1 \cdot \pi_2 \cdot H$ and $\pi_2 \cdot \pi_1 \cdot H$.
\end{proof}

\subsubsection*{Separating Exact-Order and Order-Dependent types}
\label{sec:app}
Order-dependant are closely related to \emph{exact-order} types~\cite{help15}.
Roughly speaking, exact-order types are types in which operations are non-commutative, that is, switching the order of two operations changes the response of future operations. \cite{help15} proved that an exact-order type precludes a linearizable wait-free and linearization-help free implementation using read, write, compare-and-swap and fetch-and-add primitives. We next prove that order-dependant type and exact-order type are incomparable.
For convenience, we present the definition and results for exact-order type as appears in \cite{help15}. For more details, we refer the reader to \cite{help15}.

\proofparagraph{Notation}
Let $S$ be a sequential history. We denote by $S(n)$ the first $n$ operations in $S$, and by $S_n$ the $n$-th operation in $S$. We denote by $(S + \pi?)$ all sequential histories that contains $S$ and possibly also the operation $\pi$. In other words, $(S + \pi?)$ is a set of sequential histories that contains $S$, and also all sequential histories that are similar to $S$, except that a single operation $\pi$ is inserted in somewhere between (or before or after) the operations of $S$.

\begin{definition} [\textbf{Exact-Order Type.} Definition 4.1 \cite{help15}]
    \label{def:exact-order-type}
    An exact order type $\tau$ is a type for which there exists an operation $\pi$, an infinite sequential history of operations $W$, and a (finite or an infinite) sequential history of operations $R$, such that for every integer $n \geq 0$ there exists an integer $m \geq 1$, such that for at least one operation in $R(m)$, the result it returns in any history in $W(n+1) \cdot (R(m)+\pi?)$ differs from the result it returns in any history in $W(n) \cdot \pi \cdot (R(m)+W_{n+1}?)$.
\end{definition}

\begin{theorem} [Theorem 4.18 \cite{help15}]
    \label{the:help15}
    A wait-free linearizable implementation of an exact order type using read, write, compare-and-swap and fetch-and-add primitives cannot be linearization-help free.
\end{theorem}

\paragraph{Order-Dependent does not imply Exact-Order}

A \drr\ type is a variant of a multi-reader-multi-writer register supporting exactly two writes. The sequential specification is as follows: a \writeop($v$) operation returns \ack; a \readop\ operation returns the pair of values written by the first two \writeop\ operations in the history (if such exists, otherwise a unique symbol $\bot$ is used). In other words, all \writeop\ operations after the first two writes do not effect any other operation.

\proofparagraph{Order-Dependent Type}
To satisfy the conditions of Definition \ref{def:order-dependent}, we define $H$ to be an infinite sequence of \readop\ operations, and $\pi_1,\pi_2$ to be $\writeop(1), \writeop(2)$ operations, respectively. The proof is simple and straightforward, and is left to the reader.

\proofparagraph{Exact-Order Type}
%In a Double-Write-Register, any two operations $\pi_1, \pi_2$ such that at least one is a read operation are commutative, i.e., performing $\pi_1 \cdot \pi_2$ or $\pi_2 \cdot \pi_1$ starting from any configuration does not effect any future operation response.
An exact-order type requires an infinite sequence of operations $W$, and an operation $\pi$, such that for any $i$, $\pi$ and $W_i$ are non-commutative after $W(i-1)$. However, this implies any operation in $W$ must be a \writeop\ operation, since \readop\ operation does not effect any other operation. As a result, after the history $W_1 \cdot W_2$ the value of the register is set, and no operation effects the result of any future operation. This contradicts the requirement of $\pi$ and $W_3$ being non-commutative after the history $W_1 \cdot W_2$.
Hence, no such sequence of operations $W$ exists, and a \drr\ is not an exact-order type.

%\myparagraph{Remark}
\paragraph{Exact-Order does not imply Order-Dependent}

We prove that the \texttt{Double-Read-And-Fix} object type presented in Section \tsref{sec:LH-SL} is an exact-order type but not an order-dependent type. By Claim~\ref{cl:DRAF} there exists an implementation of a \draf\ object using registers and \CAS\ only such that it is wait-free, linearizable and has both linearization-helping and universal-helping in the crash-stop model, while it is strict-linearizable in the crash-recovery model.
This proves that Lemma \ref{lm:oo-eq} does not apply to exact-order types. Notice that one can find a type for which Lemma \ref{lm:oo-eq} can be applied, although it is not order-dependent. Nonetheless, as \draf\ type demonstrate, defining the exact set of objects for which Lemma \ref{lm:oo-eq} can be applied is not a trivial task, as known sets of objects are not suffice. 

\proofparagraph{Exact-Order Type}
Let $W$ be an infinite sequence of alternating \writeop$(0)$ and \writeop$(1)$ operations. Let $\pi$ be a \writeop$(2)$ operation, and $R = \sigma$ a single \readop\ operation.
Consider a finite prefix of $W$, denoted by $H$. w.l.o.g. assume the last operation in $H$ is \writeop$(0)$ (the other case is symmetric). That is, the next operation in $W$ is \writeop$(1)$.

The value of the \draf\ object after $H$ is $\langle 0, 1 \rangle$.
Thus, $\sigma$ after $H \cdot \writeop(1)$ returns $\langle 1,0 \rangle$,
while $\sigma$ after $H \cdot \writeop(1) \cdot \pi$ returns $\langle 2,1 \rangle$.
On the other hand, $\sigma$ after $H \cdot \pi$ returns $\langle 2,0 \rangle$,
while $\sigma$ after $H \cdot \pi \cdot \writeop(1)$ returns $\langle 1,2 \rangle$.
In all cases, the response of $\sigma$ is different when extending $H$ by \writeop$(1)$ and when extending it by $\pi$.
This proves \draf\ is an exact-order type.

\proofparagraph{Order-Dependent Type}
Assume towards a contradiction \draf\ is an order-dependant type. Let $H$ be an infinite sequence of operations, $\pi_1, \pi_2$ two operations, as in Definition \ref{def:order-dependent}.
If $H$ contains a \readop\ operation, then adding either $\pi_1$ or $\pi_2$ after the first \readop\ does not effect any operation, thus contradicting condition 1 of Definition \ref{def:order-dependent}.
Otherwise, $H$ is a sequence of \writeop\ operations. However, \writeop\ always returns $ack$, thus condition 2 of Definition \ref{def:order-dependent} does not hold.
This is a contradiction.
Hence, \draf\ is not an order-dependent type.

%\ohad{Note: to add difference from other known collections: \\
%perturable - נbounded counter with read operation \\
%exact order type - very involved, but there exists an exact order type that is not %order-dependent. It has a strict-linearizable implementation which has universal-helping.
%}

\subsubsection*{Order-dependent type plus Strict-Linearizability implies Universal-help Freedom}
\begin{lemma}
\label{lm:oo-eq}
    Let $A$ be a non-blocking implementation of an order-dependent total type $\tau$, such that $A$ is strict-linearizable in the system-wide crash-recovery model. Then $A$ is linearizable and universal-help free in the crash-stop model.
\end{lemma}

\begin{proof}
    Consider an implementation $A$ as in the lemma. Notice that $A$ is linearizable in the crash-stop model, since applying strict-linearizability correctness condition to crash-free executions only implies linearizability.
    
    Assume towards a contradiction $A$ has universal-helping. Since $\tau$ is an order-dependent type there exists an infinite history $H$ and two operations $\pi_1, \pi_2$ satisfying the conditions of Definition \tsref{def:order-dependent}.
    Consider a system with 3 processes $p_1, p_2, p_3$.
    Starting from the initial configuration, let $p_1$ invoke $\pi_1$ and publish its signature, followed by $p_2$ invoke and publish $\pi_2$ signature. Denote the resulted execution by $D$. From $D$ we let $p_3$ performs the sequence of operations $H$ in a solo-run manner. Denote the resulted execution by $E$. Since $A$ is total and non-blocking such an execution exists.
    In high-level, since $A$ has universal-helping, after a finite number of operations by $p_3$ both $\pi_1, \pi_2$ are linearized. Since $A$ is strict-linearizable, adding a system-wide crash at $D$ guarantees both operations must be linearized before the crash. Order-dependant ensures both operation must be linearized in a specific order. w.l.o.g. $\pi_1$ is linearized before $\pi_2$. We then consider a similar execution, in which $p_2$ publish its signature first and crash. For the same reasons, $\pi_2$ must be linearized after $\pi_1$, that is, after the crash of $p_2$, contradicting strict-linearizability. A detailed proof follows.

    %In a high-level, Since $p_3$ performs an infinite number of operations, by universal-helping after a finite number of its operations both $\pi_1,\pi_2$ must be linearized. Since linearizing both operations is distinct from linearizing exactly one operation or no operation at all, any linearization must linearize both operations. However, $A$ is strict linearizable, thus there exists a linearization in which both $\pi_1,\pi_2$ are linearized before the crash. w.l.o.g. we assume $\pi_1$ is linearized before $\pi_2$. We next consider a similar execution in which $p_2$ publish its signature before $p_1$ and there is a system-wide crash between the two publications. Hence, in any strict-linearization of the resulted execution, $\pi_1$ must be linearized before $\pi_2$, and in particular, $\pi_2$ is linearized after the crash of $p_2$, thus contradicting strict-linearizability.
    %A detailed proof follows.
    
    By our assumption, $A$ has universal-helping, thus after a finite number of complete operations by $p_3$ both $\pi_1,\pi_2$ have linearization points, that is, in $lin(E)$ both $\pi_1$ and $\pi_2$ are linearized.
    Consider an equivalent execution in the system-wide crash-recovery model, where the system crash at $D$. Notice this does not effect the steps of $p_3$. By abuse of notation, we denote both executions by $E$. Since $A$ is strict-linearizable, there exists a linearization $slin(E)$ in which each of $\pi_1,\pi_2$ is either linearized before the crash, or not linearized at all.
    By \ref{def:od-adding-distinct} both operations must be linearized in $slin(E)$, since they affected some operation in $H$. In other words, there exists no linearization of $E$ such that exactly one of $\pi_1,\pi_2$, or none of them, is linearized.
    Following \ref{def:od-order-distinct}, there exists a single linearization order for $\pi_1,\pi_2$ in $slin(E)$. This is so since a different order effects the response of some operation in $H$.  Assume w.l.o.g. $slin(E)= \pi_1 \cdot \pi_2 \cdot H$.
    
    Crashes are not visible to processes, thus we can construct a similar execution as follows: first $p_2$ invokes and publishes the signature of $\pi_2$, followed by a system crash. Then $p_1$ invokes $\pi_1$ and publishes its signature, followed by a solo run of $p_3$ executing operations according to $H$. Since $p_3$ is executing the same steps as in $E$, and following the same argument as above, in the resulted execution $\pi_2$ must be linearized after $\pi_1$ in \emph{any} linearization function. However, this implies $\pi_2$ must be linearized after the crash of $p_2$, thus contradicting strict-linearizability.

\end{proof}

%\ohad{to add a note - Lemma \ref{lm:oo-eq} holds also for the old identifier model}
%\textcolor{blue}{Also include a note to prove the corollary. Also the result for the individual crash-recovery model}
%
\begin{corollary}
\label{cor:oo-not-strict}
Let $A$ be a non-blocking implementation of an order-dependent total object $\tau$ such that $A$ has universal-helping in the crash-stop model. Then $A$ is not strict-linearizable in the system-wide crash-recovery model.
\end{corollary}

\begin{remark}
    An implementation $I$ that is not strict-linearizable in the system-wide new identifiers crash-recovery model imply $I$ is also not strict-linearizable in all other crash-recovery models - individual (system-wide) new (old) identifiers crash-recovery model.
    Therefore, Lemma \tsref{lm:oo-eq} and Corollary \tsref{cor:oo-not-strict} holds for all crash-recovery models.

    %The set of admissible executions of an implementation $A$ under the new identifiers system-wide crash-recovery model is a subset of the admissible executions of $A$ under any other crash-recovery model. Thus, if $A$ is not strict-linearizable in the new identifiers system-wide crash-recovery model it is also not strict-linearizable in any other crash-recovery model, and Lemma \tsref{lm:oo-eq} and Corollary \tsref{cor:oo-not-strict} holds for all crash-recovery models.
\end{remark}

\section{Strict-Linearizability vs. Valency-helping}
\label{sec:SL-VH}
%

%\textcolor{red}{How important is to state the liveness property here? Does the fact that is it a wait-free implementation that gives us the result have deeper implications? The liveness implications must be discussed in the last section in any case}

In this section we prove that strict-linearizability and valency-helping are not bonds together. That is, an implementation being strict-linearizable in the crash-recovery model is independent of it having valency-helping in the crash-stop model.
We note that we prove a slightly stronger result -- strict-linearizability can co-exist with valency-helping and linearization-helping, or with none of the helping definitions.

\begin{claim}
\label{cl:snapshot}
    There exists a strict-linearizable wait-free implementation $A$ of an object $\tau$ in the individual crash-recovery model, such that $A$ is linearizable and satisfies both linearization-helping and valency-helping in the crash-stop model.
\end{claim}

\newcommand{\segment}{\textsf{Segment}}
\newcommand{\ts}{\textsf{ts}}
\newcommand{\data}{\textsf{data}}
\newcommand{\view}{\textsf{view}}

    The proof of the lemma involves constructing a modified version of the wait-free \snapshot\ implementation of Attiya et al.~\cite{AADGMS93}. For simplicity of presentation, we present the original snapshot implementation of \cite{AADGMS93} and then discuss the required modifications. The code appears in Algorithm~\tsref{alg:snapshot}.
    %Due to space constraints, we leave the proof constructions to Appendix~\tsref{sec:app3}.
    %We now present the proof arguments for snapshot implementation presented in Algorithm~\tsref{alg:snapshot}.
    
    \proofparagraph{Modifications}
    (i) Each shared register contains the entire history written to it. That is, whenever a process writes to a register \segment\ it simply appends the new value to the content of the register. Therefore, $\segment[i]$ contains a list of triplets  $\langle \ts, \data, \view \rangle$ written to it along the execution. Since all registers in the implementation are single-writer, this can be done in a safe and atomic manner.
    (ii) Any reference in the code to $\segment[j]$ simply refers to the last triplet in it, except for the following case:
    a \scan\ operation returning an indirect view in line~\ref{snapshot-undirected} chooses $j$ to be the lowest index for which $b[j].\ts-c[j].\ts \geq 2$, and it chooses the first view in $b[j]$ satisfying the former inequality (the first triplet in $b[j]$ satisfying it).
    
    The correctness proof for the new variant is similar to the proof of the original algorithm for both linearizability and wait-freedom and is left to the reader.
    %For lack of space, we leave to proof to the reader.
    %In the following we assume all initial values are 0.
    %We next prove the implementation is also strict-linearizable.
    
%\end{proof}

\begin{algorithm}[]
        
        \footnotesize
        
        \begin{multicols*}{2}
        
		\begin{flushleft}
		    \textbf{Shared variables:} \\
		        $\qquad \segment[i].\ts := 0$ \\ 
		        $\qquad \segment[i].\data := v_i$ \\
			    $\qquad \segment[i].\view := \langle v_0, \ldots, v_{n-1} \rangle$ \\
		\end{flushleft}
    
        \begin{procedure}[H]
			\caption{() \small $\update_{i}$ ($S, d$)}
			
			%$Segment[i] := \langle Segment[i].ts + 1, d, Segment[i].view \rangle$ \;
			$view := \scan()$ \;
			%$Segment[i].view := view$ \;
			$\segment[i] := \langle \segment[i].\ts + 1, d, view \rangle$ \;
		\end{procedure}
		
		\columnbreak
		
		\begin{procedure}[H]
			\caption{() \small $\scan_{i}$ ($S$)}
			
			\lFor {all $j \neq i$} {$c[j] := \segment[j]$}
			\While {\True} {
			   \lFor {all $j$} {$a[j] := \segment[j]$}
			   \lFor {all $j$} {$b[j] := \segment[j]$}
			   \uIf {for some $j \neq i$: $b[j].\ts - c[j].\ts \geq 2$} {
			        \KwRet $b[j].\view$ \label{snapshot-undirected}
			   }
			   \uElseIf {for all $j$: $a[j]=b[j]$} {
		            \KwRet $\langle b[0].\data, \ldots, b[n-1].\data \rangle$
			   }
			}
		\end{procedure}
		
		\end{multicols*}
		
		\caption{\texttt{Atomic Snapshot}. Code for process $p_i$}
		\label{alg:snapshot}
\end{algorithm}

\begin{proof}

    In the following we assume all initial values are 0.
    %We next prove the implementation is also strict-linearizable.
    
    \proofparagraph{Strict-Linearizability}
    \scan\ is a read-only operation, and thus clearly satisfy strict-linearizability -- a crashed \scan\ is considered to not have a linearization point.
    An \update\ operation is linearized when $p_i$ writes to $\segment[i]$. Hence, if a crash occurs before this write then the operation has no linearization point in any extension,
    otherwise it has a linearization point before the crash.
    
    \proofparagraph{Linearization-helping}
    Consider the following execution. Process $p_3$ invokes a \scan\ operation, collects the copy of $a$, starts collecting the copy of $b$ by reading $\segment[1]$, and halts. Process $p_1$ then completes $\update_{1}$ ($S,1$) operation, followed by $p_2$ invocation of $\update_{2}$ ($S,2$), where $p_2$ completes the \scan\ and halts just before writing to $\segment[2]$. Denote the resulted execution by $E$.
    
    Starting from $E$, a solo-run of $p_3$ completes collecting the copy $b$, and since $a = b$ (all values are the initial values), $p_3$ returns $\langle 0, 0, 0 \rangle$. This implies the \scan\ of $p_1$ is linearized before $\update_{1}$ of $p_1$.
    On the other hand, if $p_2$ takes its next step after $E$ and writes to $\segment[2]$, then any future \update\ operation will observe the $\update_{1}$ ($S,1$) operation. Thus, in any extension, $p_3$'s \scan\ will return a snapshot which observe $\update_{1}$ ($S,1$) or a later update by $p_1$. Hence, in any extension the linearization point of $\update_{1}$ ($S,1$) is before the \scan\ of $p_3$, i.e., $p_1$'s $\update_{1}$ operation is decided before $p_3$'s \scan\ operation due to a step by $p_2$. This proves linearization-helping.
    
    \proofparagraph{Valency-helping}
    To prove valency-helping, consider the following execution. Process $p_3$ executes a \scan\ operation and stops after collecting the first copy of $a$. Process $p_1$ then completes an $\update_{1}$ ($S,1$) operation, starts a new $\update_{1}$ ($S,2$) operation in which it completes the \scan\ and halts just before the write to $\segment[1]$. Denote the resulted execution by $E$, and the next step of $p_1$ by $e$.
    
    A step by $p_1$ results the execution $E \cdot e$, such that in any extension of it $p_3$'s \scan\ returns the view written by $p_1$ in the step $e$, that is, the view $\langle 1, 0, 0 \rangle$. The timestamp associated with this view is greater by 2 then then timestamp $p_1$ read and stored in $c[1]$. Following the selection rule for indirect view, even in case more \update\ operations are executed after this point, $p_3$ will always choose to return this view. This implies $p_3$'s \scan\ is univalent after the step $e$.
    On the other hand, if starting from $E$ $p_2$ completes an $\update_{2}$ ($S,1$) operation, followed by a solo run of $p_3$, the \scan\ operation completes and returns $\langle 1, 1, 0 \rangle$. This implies $p_3$'s \scan\ is multivalent in $E$. This proves valency-helping.
\end{proof}

\begin{remark}
    Algorithm \tsref{alg:snapshot} is also strict-linearizable in the more restricted system-wide crash-recovery model. Moreover, a careful analysis of the strict-linearizability proof suggests it also holds for the old identifiers crash-recovery model. Hence, Claim \tsref{cl:snapshot} holds for all crash-recovery models.
\end{remark}

\begin{lemma}
    \label{lm:SL-vs-VH}
    An implementation being strict-linearizable in the crash-recovery model is independent of it satisfying valency-helping in the crash-stop model.
\end{lemma}

\begin{proofsketch}
    By Claim \tsref{cl:snapshot}, there exists a strict-linearizable implementation in the crash-recovery model that has also valency-helping in the crash-stop model.
    On the other hand, take for example the well-known Harris linked-list \cite{Harris01}. It is strict-linearizable in the crash-recovery model since any operation is linearized on a step by its owner process. However, it is valency-help free in the crash-stop model since processes do not help each other (except for physical removal of nodes). We note that Harris linked-list does not satisfy any of the helping definitions.
\end{proofsketch}
\section{Discussion}
\label{sec:disc}
%
%\myparagraph{The crash-recovery model}
%\myparagraph{Other correctness conditions}
The correctness condition considered in this paper is strict-linearizability. However, weaker conditions have been proposed for the crash-recovery model, including durable-linearizability~\cite{IzraelevitzMS16} and recoverable-linearizability~\cite{golab15}.
Helping mechanisms seemingly delineate linearizability from strict-linearizability since crashed operations may be linearized anytime in the future. However, this is not the case with durable- and recoverable-linearizability: a pending operation can be linearized anytime in the future, thus intuitively, it does not preclude any kind of helping.
Moreover, any linearizable implementation in the crash-stop model is also durable- and recoverable-linearizable in the new identifiers crash-recovery model.
Correctness conditions, such as detectability~\cite{QueueFriedman18} and nesting-safe-recoverable-linearizability~\cite{NestingSafe18} has been proposed for a different model in which processes are aware to crash events. More precisely, upon recovery from a crash, processes execute a special recovery function responsible to fix any inconsistencies in the data-structure before proceeding with their normal execution.

%These conditions allows a crashed operation to be linearized after it crashed, thus does not prohibits helping mechanisms. Moreover, any linearizable implementation in the crashed-stop model is also durable-linearizable and recoverable-linearizable in the new identifiers crash-recovery model. On the other hand, strict-linearizability does not allow a crashed operation to be linearized after it crashed, hence it relation to helping mechanisms is more involved.

The model presented in this paper is an abstract model in which all writes are immediately persistent. This is useful for exploring the limitations of the crash-recovery model, and derive lower-bounds and impossibility results. However, real-world machines introduce another layer of complexity, since caches are volatile. In such machines, a data that has been written to main memory but is yet to be persisted is lost in case of a crash. Therefore, there is a need to carefully and manually regulate eviction of cache lines to main memory in order to avoid critical data loss.
Izraelevitz et al.~\cite{IzraelevitzMS16} defined and studied such a model, called explicit epoch persistency. In addition, \cite{IzraelevitzMS16} proposed a general durability transformation to transform any implementation from the abstract model to the more realistic model with volatile cache. Although the transformation has been proven to satisfy durable-linearizability, a similar transformation can be used for strict-linearizability. Therefore, our algorithmic results holds also for the explicit epoch persistency model.

%\newpage
\bibliographystyle{abbrv}
\bibliography{references}

%\appendix
%\newpage
%\input{appendix-order-dependent}

\end{document}